\documentclass[sigconf, nonacm]{acmart}

\newcommand\vldbavailabilityurl{URL_TO_YOUR_ARTIFACTS}
 
\usepackage[ruled]{algorithm2e}

\usepackage{colortbl}
\usepackage{tikz}
\usepackage{adjustbox}
\usepackage{booktabs}

\usepackage{subfigure} 
\usepackage{autobreak}
\usepackage{multirow}
\usepackage{url}
\usepackage{array}

\usepackage{enumitem}

\usepackage[normalem]{ulem}

\usepackage{tikz}
\usetikzlibrary{tikzmark}

\AtBeginDocument{%
  \providecommand\BibTeX{{%
    \normalfont B\kern-0.5em{\scshape i\kern-0.25em b}\kern-0.8em\TeX}}}

\begin{document}

\title{Coo: Consistency Check for Transactional Databases}



\author{Haixiang Li, Yuxing Chen, Xiaoyan Li} 
\email{lihaixiangDB@gmail.com;axinggu@gmail.com;li_xiaoyan@pku.edu.cn}





\begin{abstract}

In modern databases, transaction processing technology provides ACID (Atomicity, Consistency, Isolation, Durability) features. Consistency refers to the correctness of databases and is a crucial property for many applications, such as financial and banking services. However, there exist typical challenges for consistency. Theoretically, the current two definitions of consistency express quite different meanings, which are causal and sometimes controversial. 
Practically, it is notorious to check the consistency of databases, especially in terms of the verification cost. 

This paper proposes Coo, a framework to check the consistency of databases. Specifically, Coo has the following advancements. First, Coo proposes partial order pair (POP) graph, which has a better expressiveness on transaction conflicts in a schedule by considering stateful information like Commit and Abort. By POP graph with no cycle, Coo defines consistency completely. Secondly, Coo can construct inconsistent test cases based on POP cycles. These test cases can be used to check the consistency of databases in accurate (all types of anomalies), user-friendly (SQL-based test), and cost-effective (one-time checking in a few minutes) ways. 

We evaluate Coo with eleven databases, both centralized and distributed, under all supported isolation levels. The evaluation shows that databases did not completely follow the ANSI SQL standard (e.g., Oracle claimed to be serializable but appeared in some inconsistent cases), and have different implementation methods and behaviors for concurrent controls (e.g., PostgreSQL, MySQL, and SQL Server performed quite differently at Repeatable Read level). Coo aids to comprehend the gap between coarse levels, finding more detailed and complete inconsistent behaviors.

\end{abstract}





\maketitle



\section{Introduction}

Nowadays, real-world applications rely on databases for data storage, management, and computation.
Transaction processing is one of the key components to guaranteeing the consistency of data. Especially, financial industries like securities companies, banks, and e-commercial companies often have zero tolerance for the inconsistency of any data anomalies in any form for their core transaction data. 
However, there exist typical challenges for consistency, and there is no direct and simple method to guarantee or check the consistency.

\vspace{0.1cm}
\noindent \textbf{Motivation.} \indent Obtaining consistency for databases is vital yet it is known to be notorious and challenging from several perspectives. 
(i) \textbf{It lacks standards.} Theoretically, there exist two classical definitions with different meanings for consistency. These definitions are casual and consistency is guaranteed by either eliminating certain types of anomalies \cite{DBLP:conf/ds/GrayLPT76} or satisfying integrity constraints \cite{ansisql,2002Concurrency}. 
The former divides consistency into four degrees, and each degree gradually forbids four standard anomalies. This is very mature in designing 2PL \cite{DBLP:journals/cacm/EswarranGLT76} protocol and standard isolation levels \cite{ansisql}.
The latter checks consistency by verifying if the result satisfies integrity constraints. 
However, lacking the quantified standard of consistency may cause confusion or misuse of databases in production. For example, Oracle claimed the Serializable level supported in their databases by preventing all four standard anomalies yet proved to be only Snapshot Isolation level (more detailed anomalies shown in Table \ref{table:evaluation_dbtest_result}).
Practically, it requires huge effort to design a good black-box testing tool for consistency checks. This is twofold.
(ii) \textbf{It has a high knowledge bar.} 
The learning cost for users is high from setting up environments and modifying system modules, to understanding/modifying test cases and analyzing and debugging anomalies. Application scenarios are sometimes limited as some database services are closed-source or cloud-based where users often are not allowed to make changes or collect intermediate profiles. 
(iii) \textbf{It has a high verification cost.}  
Neither collecting nor checking is cost-effective \cite{DBLP:conf/rv/SumnerHD11,DBLP:journals/vldb/ZellagK14,DBLP:conf/concur/Cerone0G15,DBLP:journals/concurrency/Lowe17}. It is proved to be a NP-complete problem \cite{DBLP:journals/jacm/Papadimitriou79b,DBLP:journals/pacmpl/BiswasE19} to verify a serializable commit order of all transactions via little known information of read-from dependency from input and output profiles (e.g., Cobra \cite{DBLP:conf/osdi/TanZMW20}). 
Some excellent works by random tests (e.g., Elle \cite{Jepsen-Hermitage, DBLP:journals/pvldb/AlvaroK20}) can simulate some anomaly cases, but may waste a lot of time and computation on checking consistent transactions. Worse, the anomaly behaviors by these random tests sometimes can be hardly analyzed and reproduced.

\begin{table*}[]
\caption{A thorough survey on data anomalies in existing literature.}
 \resizebox{0.95\linewidth}{!}{
\footnotesize
\begin{tabular}{l|c|l|l}
\toprule
\hline
\textbf{No} &
  \textbf{Anomaly, reference, year} &
  \textbf{Examples or expressions in original papers} & \textbf{Our expressions} (Table \ref{table:anomaly_classification}) \\ \hline

1   &  \begin{tabular}[c]{@{}c@{}}Dirty Write \cite{ansisql} 1992  \end{tabular} &
  \begin{tabular}[c]{@{}l@{}}$W_1[x_1]...W_2[x_2]...$(($C_1$ or $ A_1$) and  ($C_2$ or $A_2$) in any order)\end{tabular} & Dirty Write \\ \hline

2   &   \begin{tabular}[c]{@{}c@{}}Lost Update \cite{10.1145/223784.223785} 1995   \end{tabular} &
   $R_1[x_0]$...$W_2[x_1]$...$C_2$...$W_1[x_2]$  & Lost Update Committed \\\hline


3   &   \begin{tabular}[c]{@{}c@{}}Dirty Read \cite{ansisql} 1992  \end{tabular} &
   $W_1[x]...R_2[x]$...($A_1$ and $C_2$ in either order) & Dirty Reads \\\hline

4   &   \begin{tabular}[c]{@{}c@{}}Aborted Read \cite{xie2015high} 2015, \cite{839388} 2000 \end{tabular} &   $W_1[x:i]...R_2[x:i]$...($A_1$ and $C_2$ in any order) & Dirty Reads \\\hline

5   &   \begin{tabular}[c]{@{}c@{}}Fuzzy/Non-repeatable Read \cite{ansisql} 1992    \end{tabular} &   $R_1[x]...W_2[x]...C_2...R_1[x]...C_1$  & Non-repeatable Read Committed   \\\hline

6   &    \begin{tabular}[c]{@{}c@{}}Phantom \cite{ansisql} 1992 \end{tabular} &   $R_1[P]$...$W_2${[}$y$ in $P${]}...$C_2$...$R_1[P]$...$C_1$  & Phantom \\\hline

7   &   \begin{tabular}[c]{@{}c@{}}Intermediate Read \cite{xie2015high} 2015, \cite{839388} 2000   \end{tabular} &   $W_1[x:i]...R_2[x:i]...W_1[x:j]$...$C_2$ & Intermediate Read \\\hline

8   &   \begin{tabular}[c]{@{}c@{}}Read Skew \cite{10.1145/223784.223785} 1995  \end{tabular} &   $R_1[x_0]...W_2[x_1]...W_2[y_1]...C_2...R_1[y_1]$  & Read Skew Committed\\\hline


9   &   \begin{tabular}[c]{@{}c@{}}Unnamed Anomaly \cite{schenkel2000federated} 2000   \end{tabular} &   \begin{tabular}[c]{@{}l@{}}$R_3[y]...R_1[x]...W_1[x]...R_1[y]...W_1[y]...C_1...R_2[x]...W_2[x]...R_2[z]...W_2[z]...C_2...R_3[z]...C_3$\end{tabular}  & Step IAT \\\hline

10   &   \begin{tabular}[c]{@{}c@{}}Fractured Read \cite{cerone2017algebraic} 2017, \cite{10.1145/2909870} 2014  \end{tabular} &
   $ R_1[x_0]...W_2[x_1]...W_2[y_1]...C_2...R_1[y_1] $  & Read Skew Committed \\\hline

11   &   \begin{tabular}[c]{@{}c@{}}Serial-concurrent-phenomenon \cite{binnig2014distributed}   2014   \end{tabular} &   $R_1[x_0]...W_2[x_1]...W_2[y_1]...C_2...R_1[y_1] $  & Read Skew Committed  \\\hline

12   &   \begin{tabular}[c]{@{}c@{}}Cross-phenomenon \cite{binnig2014distributed} 2014  \end{tabular} &   \begin{tabular}[c]{@{}l@{}}$R_1[x_0]...R_2[y_0]...W_3[x_1]...C_3... W_4[y_1]... C_4...R_2[x_1]...R_1[y_1]$\end{tabular}   & Step IAT \\\hline

13   &   \begin{tabular}[c]{@{}c@{}} Long Fork Anomaly \cite{cerone2017algebraic} 2017   \end{tabular} &   \begin{tabular}[c]{@{}l@{}}$R_4[x_0]...W_1[x_1]...R_3[y_0]... R_3[x_1]...W_2[y_1]...R_4[y_1]$\end{tabular}   &  Step RAT \\\hline

14   &   \begin{tabular}[c]{@{}c@{}}Causality Violation Anomaly \cite{cerone2017algebraic} 2017    \end{tabular} &
   \begin{tabular}[c]{@{}l@{}}$R_1[x_0]...W_2[x_1]...C_2...R_3[x_1]... W_3[y_1]...C_3...R_1[y_1]$\end{tabular}   & Step IAT \\\hline


15   &   \begin{tabular}[c]{@{}c@{}}Read-only Transaction Anomaly \cite{10.1145/1031570.1031573,Read_Only_Transactions} 2004   \end{tabular} &
   \begin{tabular}[c]{@{}l@{}}$R_2[x_0,0]...R_2[y_0,0]...R_1[y_0,0]... W_1[y_1,20]...  C_1...R_3[x_0,0]...R_3[y_1,20]...C_3...W_2[x_1,-11]... C_2$\end{tabular}   & Step IAT \\\hline

16   &   \begin{tabular}[c]{@{}c@{}}Write Skew \cite{10.1145/223784.223785} 1995   \end{tabular} &   $R_1[x_0]...R_2[y_0]...W_1[y_1]...W_2[x_1]$   & Write Skew \\\hline

17   &   \begin{tabular}[c]{@{}c@{}}Predicate-based Write Skew \cite{DBLP:journals/tods/FeketeLOOS05} 2005    \end{tabular} & $R_1[P]...R_2[P]...W_1$ [$y_1$ in $P$]$...W_2$[$x_1$ in $P$] & Predicate-based Write Skew \\\hline

18   &  \begin{tabular}[c]{@{}c@{}}Read Partial-committed \cite{duxiaoyong_read_partial_committed} 2019    \end{tabular} &  $R_1 [x]...W_2 [x]...W_2 [y]...C_2...R_1 [y]...C_1$ & Read Skew Committed \\\hline

19   &  \begin{tabular}[c]{@{}c@{}}Prefix violation \cite{burckhardt_et_al:LIPIcs:2015:5238} 2015    \end{tabular} &  $R_1 [x,1]...W_1 [x,2]...R_2 [y,1]...W_2 [y,2]...R_3 [x,2]...R_3 [y,1]...R_4 [y,2]...R_4 [x,1]$ & Step RAT \\\hline

\bottomrule
\end{tabular}

}
\label{tab:intro_data_anomaly}
\end{table*}

\indent These real-time \cite{DBLP:conf/pldi/XuBH05,DBLP:conf/icse/HammerDVT08,DBLP:conf/ipps/SinhaM10,DBLP:conf/rv/SumnerHD11,DBLP:journals/vldb/ZellagK14,DBLP:conf/popl/BrutschyD0V17,DBLP:conf/concur/NagarJ18} or post-verify  \cite{DBLP:journals/pacmpl/BiswasE19,DBLP:conf/osdi/TanZMW20,DBLP:journals/pvldb/AlvaroK20} solutions are often costly and user-side burdened. This drives us to a root-cause question that can we discover, define, and generate all forms of data anomalies so that we can feed them all into databases and cost-effectively check once and for all. To address the question, we discuss current challenges of lacking of standards from two aspects, i.e., the formal definition of (1) data anomalies and (2) consistency.

\noindent \textbf{Challenge 1. Define data anomalies.} \indent 
The ANSI SQL \cite{ansisql} specifies four isolation levels and four data anomalies.
This standard is classical and has been widely used in real databases. However, the definition of data anomalies is casual and has been controversial from time to time \cite{DBLP:conf/sigmod/BerensonBGMOO95}. 
The standard anomalies are single-object and avoided by lock-based protocols, yet more complex data anomalies, which are occasionally reported case by case as shown in Table \ref{tab:intro_data_anomaly}, are hardly fit into defined levels.
Existing literature \cite{ansisql, 10.1145/223784.223785, 839388} revised the definition to some extent. 
However, there is still little research to define and classify complete data anomalies, resulting in that the anomalies can still be ambiguous interpretations without a formal expression. For example, \textit{Long Fork Anomaly} \cite{cerone2017algebraic} and \textit{Prefix Violation} \cite{burckhardt_et_al:LIPIcs:2015:5238} have the same expression yet reported by different instances. Many deadlocks (e.g., \cite{DBLP:conf/sigmod/LyuZXGWCPYGWLAY21}), both local and global, are introduced and discussed, yet we think they are also a form of anomalies.

\noindent \textbf{Challenge 2. Relate inconsistency to all data anomalies.}  

\noindent There exist two previous works defining the consistency of databases. The first by Jim Gray et al. \cite{DBLP:conf/ds/GrayLPT76} defined several levels of consistency, which are strongly related to the ANSI SQL standard anomalies and 2PL protocol \cite{Gray1976Granularity}. 
The second \cite{ansisql,2002Concurrency} defined the consistency such that the final result is the same as one of the serializable schedules. 
However, both definitions hardly correlate with newly reported or undiscovered anomalies.
For example, new anomalies like \textit{Full Write Skew} (in Table \ref{table:anomaly_classification}) are hard to be quantified into previous definitions and their levels. Not to mention that with slightly different schedules (e.g., \textit{Non-repeatable Read} and \textit{Non-repeatable Read Committed}) may behave quite differently between databases (result in Table \ref{table:evaluation_dbtest_result}). Lacking complete mapping between data anomalies and inconsistency may lead to incomplete and sometimes non-reproducible consistency check (e.g., Elle \cite{Jepsen-Hermitage, DBLP:journals/pvldb/AlvaroK20}).

\begin{figure}[t]
  \centering
  \includegraphics[width=0.80\linewidth]{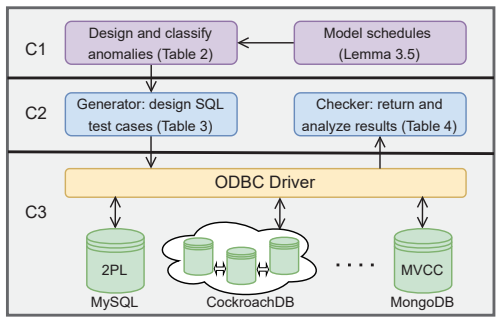}
  \caption{Coo framework. Contributions are C1: theoretical basis, C2: consistency check modules, and C3: evaluation and analysis of eleven databases.}
  \label{figue:coo_framework}
\end{figure}

\vspace{0.1cm}
\noindent \textbf{Contribution (C).} \indent
This paper proposes Coo, which contributes to pre-check the consistency of databases, filling the gap in contrast to real-time or post-verify solutions. Figure \ref{figue:coo_framework} shows the framework of Coo, which contributes to the following three aspects:
\begin{itemize}[leftmargin=*]
    \item \textbf{C1: Coo has theoretical basics.}  We propose \textit{Partial Order Pair (POP) Graph}, considering stateful information (i.e., commit, and abort), can model any schedule, compared to the traditional conflict graph which is limited to model transaction history. For example, we will show that \textit{Read Skew} (without stateful information) and \textit{Read Skew Committed} (with stateful information), which were treated as the same previously, are completely different anomalies (i.e., different formal expressions in Table \ref{table:anomaly_classification} and different evaluation behaviors in Table \ref{table:evaluation_test_case_read_skew}). By POP cycles, we can define all data anomalies, correlating reported known (e.g., Dirty Read and Deadlocks) and unexposed mysterious anomalies. 
    \item \textbf{C2: Coo is black-box and cost-effective.} The core consistency check modules are a generator and a checker, which are independent of databases. The generator provides SQL-like quires and schedules based on our definition of data anomalies, and the checker recognizes the consistent and inconsistent behaviors of the executed schedules. Each defined anomaly will be tested individually by issuing parallel transactions via ODBC driver to tested databases. The consistency check is accurate (all types of anomalies), user-friendly (SQL-based test), and cost-effective (one-time checking in a few minutes).
    \item \textbf{C3: Eleven databases are evaluated.} Through the evaluation, both centralized and distributed, we unravel the consistent and inconsistent behaviors under different isolation levels. And we are, as far as our knowledge goes, the first to propose methods for distributed evaluation. 
     We can specifically show the occurrence of anomaly types in non-consistent databases or at their weaker isolation levels. 
    Our evaluation found that some databases (e.g., Oracle, OceanBase, and Greenplum) claimed to be serializable but can not avoid some \textit{IAT} anomalies (defined in Section \ref{sec:consistency_check_in_practice}). Also,  we in-depth analyze the behaviors of different databases at different isolation levels with various implementation methods (e.g., different behaviors to designed anomaly cases by PostgreSQL, MySQL, and SQL Server in Repeatable Read level).
\end{itemize}

The rest of this paper is organized as follows.
Section \ref{sec:preliminary} presents the preliminary. 
Section \ref{sec:anoamly_model} introduces our new model to define data anomalies and correlate inconsistency.
Section \ref{sec:evaluation} evaluates our model with real databases.
Section \ref{sec:related_work} surveys the related work.
Section \ref{sec:conclusion} concludes the paper.

\section{Preliminary}\label{sec:preliminary}
This section provides the preliminary that will be used and extended in the following section. 

\paragraph{Objects, Operations, Transactions}
We consider storing data \textbf{objects} $Obj=\{x,y,...\}$ in a database.  Operations are divided into two groups, i.e., object-oriented operations and state-expressed operations. \textbf{Object-oriented operations} are operations on objects by reading or writing. Let $Op_i$ describe the possible invocations set: reading or writing an object by transaction $T_i$.  \textbf{State-expressed operations} are operations to express states of transactions, consisting of Commit (C) and Abort (A). \textbf{Transaction} is a group of operations, interacting objects, with or without a state-expressed operation at the end, representing a committed or an active state. 
We use subscripts to represent the transaction number. For example, $Op_i[x_n]$ is $x$-oriented operations by transaction $T_i$; $C_i$ and $A_j$ are the committed and abort operations by $T_i$ and $T_j$, respectively.

\paragraph{Schedules} An Adya \cite{adya1999weak} history $H$ comprises a set of transactions $T$ on objects, an order $E$ over operations $Op$ in $T$. The $E$ is persevered the order within a transaction and obeyed the object version order $<_s$. A \textbf{schedule} $S$ is a prefix of $H$.



\begin{example}\label{exp:schedule} We show an example of a schedule $S_1$ in the following:
\begin{equation}
S_1 = R_1[x_0]~R_3[x_0]~W_1[y_1]~R_3[y_1]~C_3~W_2[x_1]~R_1[y_1]~A_1.
\end{equation}
which involves three transactions, where $T_1 = R_1[x]W_1[y]R_1[y]A_1$, $T_2 = W_2[x]$, and $T_3 = R_3[x]R_3[y]C_3$ are aborted, active, and committed transactions respectively.
The set of operations is $Op(S_1)=\{R_1[x], R_3[x], W_1[y], R_3[y], W_2[x], R_1[y]\}$. For operations on the same object, we have the version order, e.g., $R_1[x_0]<_sW_2[x_1]$. Note we don't have version order between two reads, e.g., ($R_1[x_0]$, $R_3[x_0]$) or between different objects, e.g, ($R_3[x_0]$, $W_1[y_1]$), meaning reversing these operations may be an \textit{equivalent} schedule.
\end{example}

\paragraph{Conflict dependency and Conflict graph}
Every history is associated with a conflict graph (also called directed serialization graph) \cite{DBLP:journals/tods/BernsteinG83,DBLP:conf/eurosys/YabandehF12}, where nodes are committed transactions and edges are the conflicts (read-write, write-write, or write-read) between transactions.  
The conflict graph is used to test if a schedule is serializable. Intuitively, an acyclic conflict graph indicates a serializable schedule, thus the consistent execution and final state. Figure \ref{fig:compare_ser_popg}(a) depicts the graphic representation of $S_1$.

\section{Consistency Model}\label{sec:anoamly_model}
This section introduces a new consistency model called Coo that can correlate all data anomalies.  Specifically, we first proposed \textit{Partial Order Pair (POP) Graph}, which also considers state-expressed operations.  We then show any schedule can be represented by a POP graph and our checker can check an anomaly via its POP cycle.  Lastly, our generator constructs both centralized and distributed test cases based on POP cycles for the evaluation.

\subsection{Partial Order Pair Graph}\label{sec:Definition_pop_graph}

Adya's model introduced some non-cycle anomalies \cite{adya1999weak,DBLP:journals/pvldb/AlvaroK20} like Dirty Reads and Dirty Write. The reason is that they did not consider state-expressed operations in conflict graph, yet these operations sometimes may be equivalent to object-oriented ones \cite{DBLP:conf/podc/CrooksPAC17}. We strive to map all anomalies via cycles by considering these state-expressed operations. We first formally define POPs as extended conflicts in the following.

\begin{definition}\textsc{\textbf{Partial Order Pair (POP).}}\label{def:pop}
Let $T_i$, $T_j$ be transactions in a Schedule $S$ and $T_i \neq T_j$.  A Partial Order Pair (POP) is the combination of object-oriented and state-expressed operations from $T_i$ and $T_j$ and satisfies:
\begin{itemize}
    \item both transactions operate on the same object;
    \item at least one operation affects the object version (a write or a rollback of a write).
\end{itemize}
\end{definition}

\begin{lemma}\label{lemma:pop_types}
There exist at most 9 POPs in an arbitrary schedule, i.e.,$POP =  \{WW, WR, RW, WCW, WCR, RCW, RA, WC, WA\}$. 
\end{lemma}
\begin{proof}
The proof can be trivially achieved by enumerating all possible combinations of object-oriented and state-expressed operations.
Let $T_i, T_j$ be transactions in a Schedule $S$ and $p_i \in T_i$ with $q_j \in T_j $ being object-oriented operations that access the same object, $(p_i, q_j)\in \{W_iW_j, W_iR_j,R_iW_j\}$. The following is a list of all possible combinations.
\begin{itemize}
    \item[1.] $p_i-q_j$: Both transactions $T_i$ and $T_j$ are still active.
\end{itemize}
The transaction $T_i$ ends before $T_j$:
\begin{itemize}
    \item[2.] $p_i-C_i-q_j$: $T_i$ commits before $q_j$;
    \item[3.] $p_i-A_i-q_j$: $T_i$ aborts before $q_j$;
    \item[4.] $p_i-q_j-C_i$: $T_i$ commits after $q_j$;
    \item[5.] $p_i-q_j-A_i$: $T_i$ aborts after $q_j$;
\end{itemize}
The transaction $T_i$ ends after $T_j$:
\begin{itemize}
\item[6.] $p_i-q_j-C_j$: $T_j$ commits after $p_i$;
\item[7.] $p_i-q_j-A_j$: $T_j$ aborts after $p_i$.
\end{itemize} 
The operation $p_i$ will not affect the operation $q_j$ in combination $3$ due to the timely rollback of $T_i$. So does combination $7$. We obtain $15$ cases by substituting $\{W_iW_j, R_iW_j, W_iR_j\}$ into $(p_iq_j)$ of the remaining $5$ combinations.

Among them, $W_i W_j C_j$  and $W_i W_j$ both have the identical effect of modifying the accessing object by $W_j$, we group them together as POP $WW$. 
Similarly, we use POP $WR$ to represent $W_i R_j$ and $W_i R_j C_i$ and POP $RW$ to represent $R_i W_j$ and $R_i W_j C_j$.
Because read operations are not affected by a commit or abort, we put $R_i W_j A_i$ and $R_i W_j C_i$ into $RW$. Similarly, we put $W_i R_j C_j$ into $WR$.
Three cases with committed of $T_i$, i.e., $W_i C_i R_j[x]$, $W_i C_i W_j[x]$, and $R_i C_i W_j[x]$, are specified as types $WCR$, $WCW$, and $RCW$, respectively.

Finally, we have three special combination cases, i.e., ${W_i R_j A_i}$, ${W_i W_j C_i}$, and ${W_i W_j A_i}$, that are more complex as they have two version changing states. As for $W_i R_j A_i$, we have first changing state by $W_i R_j$ then second changing state by ${R_j A_i}$. $W_i R_j$ belongs to POP $WR$ and ${R_j A_i[x]}$ belongs to new POP $RA$. Likewise, ${W_i W_j C_i}$ has ${WW}$ and ${WC}$ POPs, and ${W_i W_j A_i}$ has ${WW}$ and ${WA}$ POPs.

In summary, these $15$ combination cases are grouped into 9 types POPs, i.e., $WW, WR, RW, WCW, WCR, RCW, RA, WC, WA$.
\end{proof}

Note that RA, WA, and WC are from the combination of a cycle, meaning RA, WA, and WC existed only when the cycle already existed, and this cycle is a 2-transaction cycle on a single object. 
Let $\mathcal{F}: POP(S) \rightarrow T(S) \times T(S)$ be the map between POPs and the transaction orders, e.g., $\mathcal{F}(W_i C_i R_j [x]) = (T_i, T_j)$.
In terms of POPs and their orders, we can define POP graphs.

\begin{definition}\textsc{\textbf{Partial Order Pair Graph (POP graph).}}\label{exmaple:cg_vs_popg}
  Let $S$ be a schedule. 
  A graph $G(S)=(V,E)$ is called Partial Order Pair Graph (POP graph), if vertices are transactions in $S$ and edges are orders in POPs derived from $S$, i.e (i) $V = T(S)$; (ii) $E = \mathcal{F}(POP(S))$.
\end{definition} 

Conflict and POP graphs differ in edges and expressiveness. Example \ref{exp:pop} exemplifies the distinction between them.

\begin{example}\label{exp:pop}
Continuing Example \ref{exp:schedule}, we obtain objects $Obj=\{x,$ $y\}$, and operations
$Op[x]=\{R_1[x_0]R_3[x_0]C_3W_2[x_1]\}$ and
$Op[y]=\{W_1[y_1]R_3[y_1]C_3R_1[y_1]A_1\}$ from $S_1$. Note that we don't put $A_1$ in $Op[x]$ as they don't have a write on object $x$ by $T_1$.
We derive POP from these operations, i.e. $\{R_1W_2[x],R_3C_3W_2[x],$ $W_1R_3[y],$ $R_3A_1[y]\}$.
The Conflict graph and the POP graph for $S_1$ are shown in Figure \ref{fig:compare_ser_popg}. Note that edges from $T_3$ to $T_2$ are different in conflict (RW) and POP (RCW) graph. This time, by a POP graph, the Dirty Read is expressed by a cycle formed by $T_1$ and $T_3$. 
\begin{figure}[t]
    \centering
    \includegraphics[width= \linewidth]{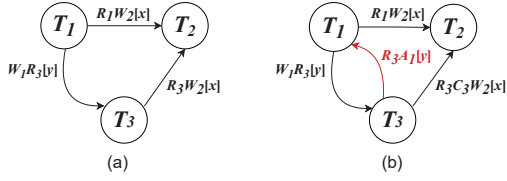}
    \caption{Comparison of (a) conflict and (b) POP graphs.}
    \label{fig:compare_ser_popg}
\end{figure}
\end{example}

\begin{lemma}
Arbitrary schedules can be represented by POP graphs. 
\end{lemma}
\begin{proof}
Given an arbitrary schedule $S$ with $Op(S)$ being the set of operations by transactions $\mathcal{T} = \{T_1, T_2, \dots, T_n\}$. First, we can derive sets of operations for variables from $S$, $\{OP[x]| x \in Obj(S)\}$. Then we can find all the combination cases in each object operation set $Op[x]$. Finally, we classify them into POPs referred to the proof of Lemma \ref{lemma:pop_types}. Through the above method, we can get the POP set $POP(S)$ corresponding to the schedule $S$. 
Then, by $\mathcal{F}$, we get the ordering between transactions based on POPs. We can model POP graphs using the transactions set and the dependent orders between transactions.
\end{proof}

\subsection{Consistency and Consistency Check}\label{sec:definition_consistency}

With POP cycles, we now are ready to define data anomalies, then define consistency with no data anomaly. 
\begin{definition}\textsc{\textbf{Data Anomaly}.}\label{def:anomaly}
The schedule exists a data anomaly exists if the represented POP graph has a cycle.
\end{definition}
The definition of data anomalies by POP graphs differs from conflict graph one in three aspects. Firstly, POP graphs model schedules instead of histories (e.g., Full Write in Table \ref{table:anomaly_classification}). Secondly, POP graphs can express all anomalies with state-expressed (e.g., Dirty Read in Definition \ref{exp:pop}). 
Thirdly, POP graphs can model more distinct anomalies (e.g., Read Skew and Read Skew Committed in Table \ref{table:anomaly_classification} are different but considered as the same by conflict graph).
We now define the consistency of a schedule.


\begin{definition}\textsc{\textbf{Consistency}}\label{def:Consistency}
Schedule $S$ satisfies consistency if the represented POP graph exists no cycle.
\end{definition}

\paragraph{Checker} 
By definition \ref{def:Consistency}, consistency, no data anomalies, and acyclic POP graphs are equivalent. Likewise, inconsistency, existing data anomalies, and existing POP cycles are equivalent. So a \textbf{consistency checker} is to test if a schedule exists a data anomaly, i.e., if the represented graph has a cycle.
In theory, the consistency check is \textbf{sound}: if it reports an anomaly in a schedule, then that anomaly should exist in every history of that schedule. 
The consistency check is \textbf{complete}: if it reports an anomaly in a schedule, then a POP cycle exists in the schedule of that anomaly. As a schedule is a prefix of history, the anomaly occurring in the schedule also occurs in the corresponding histories. So the soundness is correct. As we defined that the anomaly schedule exists a POP cycle, the completeness is also correct.

\subsection{Consistency Check in Practice}\label{sec:consistency_check_in_practice}
This part discusses the consistency check in practice. As each POP cycle may express an anomaly scenario, it is neither cost-effective nor possible to test infinite cycles. Our test cases involve trading off the cost and time spent against the completeness. We want as less as test cases to express as much as the database's inconsistent behaviors. By soundness, an anomaly may exist in different schedules or histories. We consider exploring the simplest form for a data anomaly, which will be used for the design and classification of data anomalies for the evaluation. As most known data anomalies (e.g., Dirty Write and Dirty Read) are single-object, we start with one object POP cycles.


\begin{lemma}\label{lemma:one_object_reduction}
  A POP cycle with three transactions ($N_T=3$) by one object ($N_{Obj}=1$) exists a cycle with two transactions.
\end{lemma} 

\begin{proof}
We exclude POPs RA, WA, and WC in our discussion, as these POPs appeared in a two-transaction one-object cycle, which need no proof. We first assume the POP cycle is $G = \{ \{T_1,T_2,T_3\},$ $\{(T_1,$ $T_2),$ $(T_2, T_3),(T_3, T_1)\}$. 
We let $\{(p_1,$ $q_2),$ $(p_2, q_3),(p_3, q_1)\}$ be the object-oriented operations in forming cycle $G=\{(T_1,$ $T_2),$ $(T_2, T_3),$ $(T_3, T_1)\}$. We let $<_s$ denote the version order. So the graph can be represented by $\{p_1 <_s q_2;$ $p_2 <_s q_3;$ $p_3 <_s q_1\}$. As each POP should have a write operation, we have the following situations.

If $p_1=W$, (i) if $p_1$ happens before $p_2$, i.e., $p_1 <_s p_2$, since $p_2 <_s q_3$, then $p_1 <_s q_3$, meaning a POP from $T_1$ to $T_3$. By original POP from $T_3$ to $T_1$, $T_1$ and $T_3$ forms a cycle. (ii) if $p_1$ happens later than $p_2$, i.e., $p_2 <_s p_1$, meaning a POP from $T_2$ to $T_1$, then, $T_1$ and $T_2$ forms a cycle.

If $p_1=R$, then $q_2=W$. Likewise, (i) if $q_2 <_s p_3$, then $T_1$ and $T_2$ forms a cycle. (ii) if $p_3 <_s q_2$, then $T_2$ and $T_3$ forms a cycle.
\end{proof}

\begin{lemma} \label{theorem3}
    A POP cycle with any number of transactions ($N_T\geq3$) by one object ($N_{Obj}=1$) exists a cycle with two transactions. 
\end{lemma}
\begin{proof}
The proof is by induction. The theorem holds for $N_T=3$ by Lemma
 \ref{lemma:one_object_reduction}.  We first assume theorem holds for $N_T < k$. 

When $N_T = k$, we assume the POP cycle is $G = \{ \{T_1,T_2,...,T_k\},$ $\{(T_1,$ $T_2),$ $(T_2, T_3),...,(T_k, T_1)\}$. 
We let $\{(p_1,$ $q_2),$ $(p_2, q_3),...,$ $(p_k,$ $q_1)\}$ be the object-oriented operations in forming cycle $G=\{(T_1,$ $T_2),$ $(T_2, T_3),...,(T_k, T_1)\}$. We let $<_s$ denote the version order between operations. So the graph can be represented by $\{p_1 <_s q_2;$ $p_2 <_s q_3;...;$ $p_k <_s q_1\}$. 
As each POP should have a write operation, we have the following cases.

If $p_1=W$, (i) if $p_1$ happens before $p_{k-1}$, i.e., $p_1 <_s p_{k-1}$, since $p_{k-1} <_s q_{k}$, then $p_1 <_s q_n$, meaning a POP from $T_1$ to $T_n$. By original POP from $T_k$ to $T_1$, $T_1$ and $T_k$ forms a cycle. (ii) if $p_1$ happens later than $p_{k-1}$, i.e., $p_{k-1} <_s p_1$, meaning a POP from $T_{k-1}$ to $T_1$, then, we remove $T_k$ and achieve a new cycle $G'=\{(T_1,$ $T_2),$ $(T_2, T_3),...,(T_{k-1}, T_1)\}$. Based on the assumption, when $n=k-1$ the theorem is true. 

If $p_1=R$, then $q_2=W$. Likewise, (i) if $q_2 <_s p_{k-1}$, then $T_1$, $T_2$, and $T_k$ forms a cycle. It can be reduced to 2-transaction cycle by lemma \ref{lemma:one_object_reduction}. (ii) if $p_{k-1} <_s q_2$, then, we remove $T_1$ and $T_k$, and achieve a new cycle $G'=\{(T_2,$ $T_3),$ $(T_3, T_4),...,(T_{k-1}, T_2)\}$. Based on the assumption, when $n<k$ the theorem is true.
\end{proof}

In general, if one cycle only involves one object, we can find representative cycles of exactly two transactions. This property is meaningful, as when only one object involves, evaluating two-transaction cycles is sufficient to represent cycles with more transactions.
Next, we consider a POP cycle with more than one object.

\begin{lemma} \label{lemma:at_most_two_edges_for_one_object}
    A POP cycle has more than two POPs accessing to one object exists a cycle with at most two connected POPs accessing this object. 
\end{lemma}
\begin{proof}
We first assume the POP cycle is $G = \{ \{T_1,T_2,\dots,T_n\},$ $\{(T_1, T_2),$ $(T_2, T_3),\dots,(T_n, T_1)\}$.
The POP edges accessing the same object $x$ are $\mathcal{F}(POP_i[x]) = (T_i, T_{i+1})$ and $\mathcal{F}(POP_j[x]) = (T_j, T_{j+1})$, $j >i$.
We assume $\{(p_iq_{i+1}[x]),(p_jq_{j+1}[x])\}$ are the object-oriented operations in forming edges of $POP_i[x]$ and $POP_j[x]$. 
Then $G$  can be simplified into the following graphs.

If $p_i=W$, (i) if $p_i<_s p_j$, since $p_j <_s q_{j+1}$, then $p_i <_s q_{j+1}$, meaning a POP from $T_i$ to $T_{j+1}$. 
We get  $G' = \{ \{T_1,T_2,...T_i,T_{j+1}...T_n\},$ $\{(T_1, T_2),\dots,$ $(T_i,T_{j+1}),$ $\dots$,$(T_n, T_1)\} $ with a new POP accessing $x$ edge $(T_i,$ $T_{j+1})$.
(ii) if $p_j <_s p_i$, meaning a POP from $T_j$ to $T_i$. We get $G' = \{ \{T_i,T_{i+1},$ $...T_j\},\{(T_i, T_{i+1}),\dots(T_{j-1}, T_j),(T_j, T_i\}$ with a new POP edge  $(T_j, T_i)$. The adjoining edges $(T_j, T_i)$ and $(T_i, T_{i+1})$ with ordering $p_j <_s p_i <_s <q_{i+1}$ are both accessing the same object $x$.
(ii-a) There will be no new POP edges between them until $p_j = q_j = R$, which is $\mathcal{F}^{-1}(T_j, T_i) \in \{R_jW_i, R_jC_jW_i\} $ and  $\mathcal{F}^{-1}(T_i, T_{i+1}) \in \{W_iR_j, W_iC_iR_j\}$.
(ii-b) Otherwise, meaning a POP from $p_j$ and $q_{i+1}$, causing the POP cycle to continue to be simplified to $G' = \{ \{T_i,T_{i+1},$ $...T_j\},\{(T_j, T_{i+1}),\dots(T_{j-1}, T_j)\}$ with a new POP edge $(T_j, T_{i+1})$.

If $p_i = R$, then $q_{i+1} = W$. (i) If $q_{i+1} <_s p_j$, since $p_j <_s q_{j+1}$ then $q_{i+1} <_s q_{j+1}$, meaning a POP from $T_{i+1}$ to $T_{j+1}$. We get  $G' = \{ \{T_1,T_2,...T_i,T_{i+1},T_{j+1}...T_n\},$ $\{(T_1, T_2),$ $\dots, $ $(T_{i+1},T_{j+1}),$  $\dots,(T_n, T_1)\} $ with a new POP edge $(T_{i+1},T_{j+1})$. 
The adjoining edges $(T_i, T_{i+1})$ and $(T_{i+1}, T_{j+1})$ with ordering $p_i <_s q_{i+1} <_s <q_{j+1}$ are both accessing the same object $x$.
(ii-a) If $q_{j+1} = W$, the graph $G'$ can be continues to simplify by the POP $\mathcal{F}^{-1}(T_i, T_{j+1}) \in \{R_iW_{j+1},R_iC_iW_{j+1}\}$.
(ii-b) Otherwise, if $q_{i+1} = R$, POP edges are $\mathcal{F}^{-1}(T_i, T_{i+1}) \in \{R_iW_{i+1}, $ $R_iC_iW_{i+1}\} $ and  $\mathcal{F}^{-1}(T_{i+1}, T_{j+1}) \in \{W_{i+1}R_{j+1}, W_{i+1}C_{i+1}R_{j+1}\}$.

By repeating the above steps on object $x$, we can obtain the cycle with only one or two edges operating on this object. And if two edges remained, then these two edges are connected.
\end{proof}

\begin{theorem}\label{theorem:reduction}
  A POP cycle has $N_{Obj}$ ($N_{Obj} \geq 1$) objects exists a POP cycle with at most $2N_{Obj}$ transactions.
\end{theorem}
\begin{proof}
When $N_{Obj}=1$, Lemma \ref{theorem3} has proven the theorem.

When $N_{Obj} \geq 2$, we prove it by contradiction.
Without loss of generality, we assume that there exists a cycle with $2N_{Obj}+1$ transactions and $N_{Obj}$ objects that can not be simplified. The cycle must then include three POP edges accessing the same object, e.g $x$.
However, by Lemma \ref{lemma:at_most_two_edges_for_one_object}, we can proceed to simplify the cycle to at most two POPs accessing $x$, making the original cycle at most $2N_{Obj}$ transactions, which contradicts the assumption of the simplest cycle.
\end{proof}

\begin{figure}[t]
      \centering
      \includegraphics[width= \linewidth]{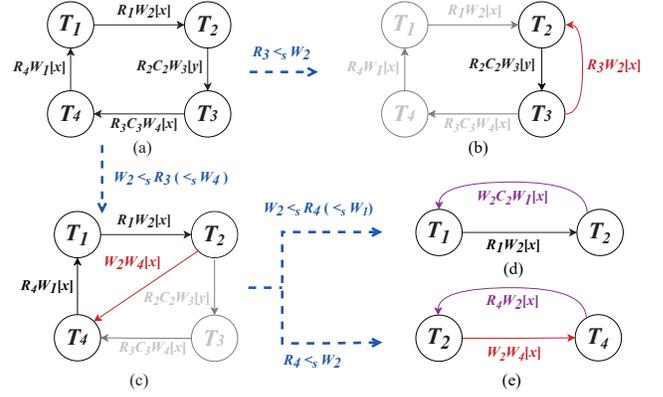}
      \caption{A 4-transaction cycle to its simplified cycles.}
      \label{theorpic4}
\end{figure}

\begin{table*}[ht]
\caption{Data anomaly formal expression, classification, and their POP combinations in POP cycles.} 
\footnotesize
\setlength\extrarowheight{1.5pt}

\begin{tabular}{c|c|c|l|l|l}
\toprule
\hline
\multicolumn{2}{c|}{\textbf{Types   of Anomalies}} & \multicolumn{1}{c|}{\textbf{No}}  & \multicolumn{1}{c|}{\textbf{Anomalies}}  & \multicolumn{1}{c|}{\textbf{Formal expressions}}                                               & \multicolumn{1}{c}{\textbf{POP   Combinations}}                                                                               \\ \hline
\multirow{12}{*}{RAT}     & SDA   & 1         & Dirty   Read \cite{ansisql,xie2015high,839388}                             & $W_i   [x_m ] \dots R_j [x_m] \dots A_i$                                                                     & $W_i R_j [x] -R_j A_i   [x] $                                                                                                    \\ \cline{2-6} 
                          & SDA    & 2         & Non-repeatable Read \cite{ansisql}                     & $R_i   [x_m ] \dots W_j [x_{m+1} ] \dots R_i [x_{m+1} ]$                                                     & $R_i W_j [x] -W_j R_i   [x] $                                                                                                    \\ \cline{2-6} 
                          & SDA    & 3         & Intermediate Read  \cite{xie2015high,839388}                        & $W_i   [x_m ] \dots R_j [x_m ]  \dots W_i [x_{m+1} ]$                                                         & $W_i R_j [x] -R_j   W_i    [x]$                                                                                 \\ \cline{2-6} 
                          & SDA     & 4        & \textbf{Intermediate Read Committed}     & $W_i   [x_m ] \dots R_j [x_m ] \dots C_j \dots W_i [x_{m+1} ]$                                               & $W_i R_j [x] -R_j C_j W_i [x] $                                                                                        \\ \cline{2-6} 
                          & SDA     & 5        & \textbf{Lost Self Update}                & $W_i   [x_m ] \dots W_j [x_{m+1} ] \dots R_i [x_{m+1} ]$                                            & $W_i W_j [x]   -W_j R_i [x] $                                                                                   \\ \cline{2-6}         
                          & DDA      & 6       & \textbf{Write-read   Skew}               & $W_i   [x_m ] \dots R_j [x_m ] \dots W_j [y_n ] \dots R_i [y_n ]$                                   & $W_i R_j [x]-W_j R_i   [y]$                                                                                             \\ \cline{2-6} 
                          & DDA       & 7     & \textbf{Write-read   Skew Committed}     & $W_i   [x_m ] \dots R_j [x_m ] \dots W_j [y_n ] \dots C_j \dots R_i [y_n ]$                         & $W_i R_j [x]-W_j C_j   R_i [y]$                                                                                         \\ \cline{2-6} 
                          & DDA      & 8      & \textbf{Double-write Skew 1}             & $W_i   [x_m ] \dots R_j [x_m ] \dots W_j [y_n ] \dots W_i [y_{n+1} ]$                               & $W_i R_j [x]-W_j W_i   [y]$                                                                                            \\ \cline{2-6} 
                          & DDA     & 9       & \textbf{Double-write Skew 1 Committed}   & $W_i   [x_m ] \dots R_j [x_m ] \dots W_j [y_n ] \dots C_j \dots W_i [y_{n+1} ]$                     & $W_i R_j [x]-W_j C_j   W_i [y]$                                                                                         \\ \cline{2-6} 
                          & DDA    & 10        & \textbf{Double-write Skew 2}             & $W_i   [x_m ] \dots W_j [x_{m+1} ] \dots W_j [y_n ] \dots R_i [y_n]$                                & $W_i W_j [x]-W_j R_i   [y]$                                                                                             \\ \cline{2-6} 
                          & DDA     & 11       & Read   Skew \cite{10.1145/223784.223785}                             & $R_i   [x_m ] \dots W_j [x_{m+1} ] \dots W_j [y_n ] \dots R_i [y_n ]$                                        & $R_i W_j [x]-W_j R_i   [y]$                                                                                                      \\ \cline{2-6} 
                          & DDA      & 12      & \textbf{Read Skew 2}                     & $W_i   [x_m ] \dots R_j [x_m ] \dots R_j [y_n ] \dots W_i [y_{n+1} ]$                               & $W_i R_j [x]-R_j W_i   [y]$                                                                                  \\ \cline{2-6} 
                          & DDA     & 13       & \textbf{Read Skew 2 Committed}           & $W_i   [x_m ] \dots R_j [x_m ] \dots R_j [y_n ] \dots C_j \dots W_i [y_{n+1} ]$                     & $W_i R_j [x]-R_j C_j W_i   [y]$                                                                                 \\ \cline{2-6} 
                          & MDA     & 14      & \textbf{Step RAT} \cite{burckhardt_et_al:LIPIcs:2015:5238,cerone2017algebraic}                       & $\dots W_{i} [x_m] \dots R_{j} [x_m] \dots $, and $N_{obj} \geq 2$,  $N_{T} \geq 3$                                            & $ \dots W_i R_j [x] \dots $                                                                                                                       \\ \hline
\multirow{11}{*}{WAT}     & SDA      & 15      & Dirty   Write \cite{ansisql}                            & $W_i   [x_m ] \dots W_j [x_{m+1} ] \dots A_i/C_i$                                                            & $W_i W_j [x] -W_j   A_i/C_i [x] $                                                                                                \\ \cline{2-6} 
                          & SDA       & 16     & \textbf{Full Write}                      & $W_i   [x_m ] \dots W_j [x_{m+1} ] \dots W_i [x_{m+2} ]$                                            & $W_i   W_j  [x] -W_j W_i [x] $                                                                                          \\ \cline{2-6} 
                          & SDA      & 17      & \textbf{Full Write Committed}            & $W_i   [x_m ] \dots W_j [x_{m+1} ] \dots C_j \dots W_i [x_{m+2} ]$                                  & $W_i W_j [x] -W_j C_j   W_i [x] $                                                                                       \\ \cline{2-6} 
                          & SDA     & 18       & Lost   Update \cite{10.1145/223784.223785}                            & $R_i   [x_m ] \dots W_j [x_{m+1} ] \dots W_i [x_{m+2} ]$                                                     & $R_i W_j [x] -W_j   W_i [x] $                                                                                            \\ \cline{2-6} 
                          & SDA    & 19        & \textbf{Lost Self Update Committed}      & $W_i   [x_m ] \dots W_j [x_{m+1} ] \dots C_j \dots R_i [x_{m+1} ]$                                  & $W_i W_j [x]   -W_j C_j R_i [x] $                                                                       \\ \cline{2-6} 
                          & DDA     & 20       & \textbf{Double-write   Skew 2 Committed} & $W_i   [x_m ] \dots W_j [x_{m+1} ] \dots W_j [y_n ] \dots C_j \dots R_i [y_n ]$                     & $W_i W_j [x]-W_j C_j   R_i [y]$                                                                                         \\ \cline{2-6} 
                          & DDA   & 21         & \textbf{Full-write   Skew}              & $W_i   [x_m ] \dots W_j [x_{m+1} ] \dots W_j [y_n ] \dots W_i [y_{n+1} ]$                           & $W_i W_j [x]-W_j W_i   [y]$                                                                                             \\ \cline{2-6} 
                          & DDA    & 22        & \textbf{Full-write   Skew Committed}     & $W_i   [x_m ] \dots W_j [x_{m+1} ] \dots W_j [y_n ] \dots  C_j \dots W_i [y_{n+1} ]$                & $W_i W_j [x]-W_j C_j   W_i [y]$                                                                                         \\ \cline{2-6} 
                          & DDA    & 23        & \textbf{Read-write   Skew 1}             & $R_i   [x_m ] \dots W_j [x_{m+1} ] \dots W_j [y_n ] \dots W_i [y_{n+1} ]$                           & $R_i W_j [x]-W_j W_i   [y]$                                                                                             \\ \cline{2-6} 
                          & DDA    & 24    & \textbf{Read-write   Skew 2}             & $W_i   [x_m ] \dots W_j [x_{m+1} ] \dots R_j [y_n ]   \dots W_i [y_{n+1} ]$            & $W_i   W_j [x]-R_j   W_i    [y]$    \\ \cline{2-6} 
                          & DDA      & 25      & \textbf{Read-write   Skew 2 Committed}             & $W_i   [x_m ] \dots W_j [x_{m+1} ] \dots R_j [y_n ]    \dots C_j \dots W_i [y_{n+1} ]$            & $W_i   W_j [x]-R_j C_j W_i [y]$    \\ \cline{2-6} 
                          & \multirow{2}{*}{MDA}    & \multirow{2}{*}{26} &  \multirow{2}{*}{\textbf{Step   WAT}}              & $\dots W_{i} [x_m] \dots W_{j} [x_{m+1}] \dots$, and $ N_{obj}\geq 2$, $N_{T} \geq 3$,                                        &  \multirow{2}{*}{$ \dots W_i W_j [x] \dots $ }                                                                                                               \\ 
                          &                         &   &                                       & and not include $( \dots {W_{i1}}[y_n ]\dots   R_{j1} [y_n] \dots )$                                   &                                                                                                                                          \\ \hline
\multirow{8}{*}{IAT}      & SDA       & 27     & Non-repeatable Read Committed \cite{ansisql} & $R_i   [x_m ] \dots W_j [x_{m+1} ] \dots C_j \dots R_i [x_{m+1} ]$                                  & $R_i W_j [x] -W_j C_j   R_i [x] $                                                                                       \\ \cline{2-6} 
                          & SDA     & 28       & \textbf{Lost   Update Committed}         & $R_i   [x_m ] \dots W_j [x_{m+1} ] \dots C_j \dots W_i [x_{m+2} ]$                                  & $R_i W_j [x] -W_j C_j   W_i [x] $                                                                           \\ \cline{2-6} 
                          & DDA     & 29       & Read Skew Committed \cite{10.1145/223784.223785}         & $R_i   [x_m ] \dots W_j [x_{m+1} ] \dots W_j [y_n ] \dots C_j \dots R_i [y_n ]$                     & $R_i W_j [x]-W_j C_j   R_i [y]$                                                                                         \\ \cline{2-6} 
                          & DDA     & 30       & \textbf{Read-write   Skew 1 Committed}   & $R_i   [x_m ] \dots W_j [x_{m+1} ] \dots W_j [y_n ] \dots C_j \dots W_i [y_{n+1} ]$                 & $R_i W_j [x]-W_j C_j   W_i [y]$                                                                                         \\ \cline{2-6} 
                          & DDA     & 31       & Write   Skew  \cite{DBLP:conf/sigmod/BerensonBGMOO95}                  & $R_i   [x_m ] \dots W_j [x_{m+1} ] \dots R_j [y_n ]     \dots W_i [y_{n+1}   ]$                     & $R_i W_j [x]-R_j   W_i    [y]$                                                \\ \cline{2-6} 
                          & DDA     & 32        & \textbf{Write   Skew Committed}            & $R_i   [x_m ] \dots W_j [x_{m+1} ] \dots R_j [y_n ]    \dots C_j  \dots W_i [y_{n+1}   ]$         & $R_i W_j [x]-R_j C_j W_i [y] $                                                                              \\ \cline{2-6}

                          & \multirow{2}{*}{MDA}    & \multirow{2}{*}{33}  & \multirow{2}{*}{\textbf{Step   IAT} \cite{schenkel2000federated,10.1145/1031570.1031573,Read_Only_Transactions,binnig2014distributed,cerone2017algebraic} }              & Not   include $( \dots W_{i1} [x_m ] \dots R_{j1} [x_m] \dots $                                            & \multirow{2}{*}{$ \dots R_i W_j [x] \dots $ }                                                                                                                      \\ 
                          &                     &    &                                          & and  $ \dots W_{i2}[y_n ] \dots W_{j2}[y_{n+1}] \dots )$, $N_{obj} \geq 2$,  $N_{T} \geq 3$                                  &                                                                                                                                      \\ \hline \bottomrule
\end{tabular}
\label{table:anomaly_classification}
\end{table*}

\begin{example}
Figure \ref{theorpic4}(a) depicts a 4-transaction POP cycle $G = \{\{T_1, T_2, $ $T_3, T_4\},\{(T_1, T_2), (T_2, T_3), $ $(T_3, T_4), (T_4, T_1)\}\} $ with $POPs = $ $\{ R_1W_2[x],R_2C_2W_3[y],R_3C_3W_4[x],$ $R_4W_1[x]\} $.
To simplify, (i) if $R_3 <_s W_2$, we obtained a new POP from $T_3$ to $T_2$, and a 2-transaction POP cycle $G'=\{(T_2,T_3),(T_3, T_2)\}$ as shown in Figure \ref{theorpic4}(b). (ii) if $W_2 <_s R_3$, then $T_1$, $T_2$, and $T_4$ forms a cycle as shown in Figure \ref{theorpic4}(c) by a new POP from $T_2$ to $T_4$. By lemma \ref{lemma:one_object_reduction}, we keep simplifying. (ii-a) if $W_2 <_s R_4$, since $R_4 <_s W_1$, then $W_2 <_s W_1$, meaning a POP from $T_2$ to $T_1$ (Figure \ref{theorpic4}(d)). (ii-b) if $R_4 <_s W_2$, then $T_2$ and $T_4$ forms a cycle (Figure \ref{theorpic4}(e)). 
\end{example}




\paragraph{Generator} 
We provide two classifications. The first is based on primitive conflict dependencies, i.e., WR, WW, and RW, i.e., 
(i) \textbf{Read Anomaly Type (RAT)}, if the cycle has at least a $WR$ POP; 
(ii) \textbf{Write Anomaly Type (WAT)}, if the cycle does not have a $WR$ POP, but have at least a $WW$ POP; 
(iii) \textbf{Intersect Anomaly Type (IAT)}, if the cycle does not have $WR$ and $WW$ POPs. 
This classification closely relates to three traditional conflicts and current knowledge, leading to a better evaluation and analysis of POP behaviors. 
So based on our classification, Read Skew ($R_1[x_0]W_2[x_1]W_2[y_1]$ $R_1[y_1]$) and Read Skew Committed (We named it) ($R_1[x_0]W_2[x_1]$ $W_2[y_1]C_2R_1[y_1]$) are different anomalies in different categories. Read Skew with WR belongs to RAT, while Read Skew Committed without WR and WW belongs IAT.

By Theorem \ref{theorem:reduction}, given finite number of transactions ($N_{T}$) and objects ($N_{obj}$), the simplified cycles are also finite and can be determinedly evaluated. This classification controls the real number of evaluation cases.
The second is based on $N_{T}$ and $N_{obj}$ in cycles, i.e.,:
(i) \textbf{Single Data Anomaly (SDA)}, if $N_{T}=2$, $N_{obj}=1$;
(ii) \textbf{Double Data Anomaly (DDA)}, if $N_{T}=2$, $N_{obj}=2$;
(iii) \textbf{Multi-transaction Data Anomaly (MDA)}, otherwise.
So the SDAs and DDAs are finite, which will be evaluated one by one, while MDAs are infinite, which will be evaluated by one of the typical cases. The four standard anomalies are SDAs.
We think this classification is sufficient to illustrate the core idea and explore relatively complete inconsistent behaviors. But we do not limit classifications with more one-to-one mapping anomalies of fixed transactions and objects for a more detailed evaluation. We also plan our future work to test databases with more random cycles by a larger number of transactions and objects.

Table \ref{table:anomaly_classification} shows all data anomalies types and their classification. The anomaly names with BOLD font are 20+ new types of anomalies that have never been reported (We named them with ``committed'' when it has a WCW, WCR, or RCW POP). Those reported in Step RAT and Step IAT are a tiny portion of them. Unlike previous tools (e.g., Elle \cite{DBLP:journals/pvldb/AlvaroK20}) which randomly issue queries and found anomaly by accident, our generator provides exact sequences of schedules (more details in Section \ref{sec:evaluation_test_case_construction}), making the consistency check determined and explainable, meaning it is easy to reproduce and to debug/analyze the result.

\begin{corollary}
  If a schedule satisfies consistency, then the schedule does not have any data anomalies in Table \ref{table:anomaly_classification}. 
\end{corollary}

The current research mainly focused on centralized databases. There is little research on distributed consistency and it remains ambiguous to do a distributed check. 
We first define distributed data anomalies.
\begin{definition}\textsc{\textbf{Distributed Data Anomalies}}\label{definition:distributed_data_anomaly}
The distributed data anomaly exists if the represented POP graph has a cycle, and it has at least two objects storing at distributed partitions.
\end{definition}  
The \textbf{distributed consistency check} is to test if a distributed data anomaly exists. The standard anomalies are not distributed ones and are insufficient for a distributed check as they are single-object.  
By our classification, we can construct a distributed data anomaly by a DDA or MDA.  
We particularly designed the test cases to access the different objects from different partitions sometimes from different tables. The design is required by table partitioning and the data is expected to insert/update in different partitions/shards (e.g., by PARTITION BY RANGE in SQL).

\section{Evaluation}\label{sec:evaluation}
In this part, we will evaluate 11 real databases with 33 designed anomaly test cases. 

\subsection{Setup}
We deployed 2 Linux machines each with 8 cores (Intel(R) Xeon(R) Gold 6133 CPU @ 2.50GHz) and 16 GB memory. The centralized evaluation only used one machine. We tested distributed OceanBase, TDSQL, and CockroachDB by their cloud services.
We installed UnixODBC for the common driver, and some database drivers are installed by the trial version connector from \textit{CData} \cite{cdata}. The tests are coded with C++. Each transaction is issued with one thread/core. 
The deadlock or wait\_die timeout is often set to 20 seconds depending on the cases. The source code is available on Github \cite{coo_consistency_check}. We execute transactions in parallel while using \textit{timesleep} (e.g., 0.1 second in centralized tests) between queries to force execution sequences. 

We evaluated eleven real databases, i.e., MySQL \cite{mysql}, MyRocks \cite{myrocks}, TDSQL \cite{tdsql}, SQL Server \cite{sqlserver}, TiDB \cite{tidb}, Oracle \cite{oracle}, OceanBase \cite{oceanbase}, Greenplum \cite{Greenplum}, PostgreSQL \cite{postgresql}, CockroachDB \cite{cockroachdb}, MongoDB \cite{mongodb}.  Most databases support four standard isolation level, i.e., Serializable (SER), Repeatable Read (RR), Read Committed (RC), and Read Uncommitted (RU). MongoDB supports only Snapshot Isolation(SI) level. Greenplum supports SER, RC and RU levels.  OceanBase support two modes, i.e., MySQL (RR and RC supported) and Oracle modes (SER, RR, and RC supported). TiDB supports RR and RC levels, as well as its Optimistic (OPT) level. SQL Server also supports two additional SI levels in optimistic mode, i.e., the default one (SI) and the read-committed snapshot level (RCSI). Table \ref{table:evaluation_dbtest_result} shows their default ("$\star$") and other supported levels. Some levels in one database perform the same, so we put them together (e.g., RC and RU in PostgreSQL). We exclude to present MyRocks and TDSQL in most cases, as they perform the same as MySQL.

\begin{table}[t!]
	\caption{PostgreSQL Evaluation by Read Skew and Read Skew Committed at the RC level and by Lost Update Committed and Step WAT at the SER level. 
	}
	\label{table:evaluation_test_case_read_skew}
	\scriptsize
\resizebox{\linewidth}{!}{
	\centering
	\begin{minipage}{\linewidth}
\begin{minipage}[t]{\textwidth}
\begin{tabular}{p{0.078in}|p{0.73in}|p{0.32in}|p{0.34in}|p{0.73in}|p{0.31in}}
\toprule
\hline
\multicolumn{6}{c}{Preparation}\\\hline
1 & \multicolumn{5}{l}{DROP TABLE IF EXISTS t1}\\\hline
2 & \multicolumn{5}{l}{CREATE TABLE t1 (k INT PRIMARY KEY, v INT)}\\\hline
3 & \multicolumn{5}{l}{INSERT INTO t1 VALUES (0, 0)}\\\hline
4 & \multicolumn{5}{l}{INSERT INTO t1 VALUES (1, 0)}\\\hline
\bottomrule 
\cellcolor{gray!30} A & \multicolumn{5}{c}{\textbf{Generator:} Read Skew ($R_1[x_0]W_2[y_1]W_2[x_1]R_1[y_1]$)}\\\hline
Q & \multicolumn{1}{c|}{Session 1: $T_1$-SQL} & \multicolumn{2}{c|}{Operations} & \multicolumn{1}{c|}{Session 2: $T_2$-SQL} & Result\\\hline
1 & Begin & & & & -\\\hline
\multirow{2}{*}{2} & SELECT * FROM t1  & \multirow{2}{*}{$R_1[x_0]$} \tikzmark{a1}   & & & \multirow{2}{*}{(0,0)} \\
 &  WHERE k=0 &  & & & \\\hline
3 & & & & Begin & - \\\hline

\multirow{2}{*}{4} & & & \tikzmark{d1} \multirow{2}{*}{$W_2[y_1]$}  & UPDATE t1 SET v=1 & \multirow{2}{*}{-} \\
 & & & & WHERE k=1 &  \\\hline

\multirow{2}{*}{5} & & & \tikzmark{b1}\multirow{2}{*}{$W_2[x_1]$}   & UPDATE t1 SET v=1  & \multirow{2}{*}{-} \\
 & & & & WHERE k=0 & \\\hline
\multirow{2}{*}{6} & SELECT * FROM t1 & \sout{$R_1[y_1]$} & & & \cellcolor{violet!20}Snapshot  \\
 & WHERE k=1 & $R_1[y_0]$ \tikzmark{c1}  & & & \cellcolor{violet!20} (1,0) \\\hline
7 & & & $C_2$ &  \cellcolor{violet!20}  Commit & -\\\hline
8 & Commit & $C_1$ & & & -\\\hline
\multicolumn{6}{c}{\cellcolor{green!16} \textbf{Checker:} Pass (\textbf{P}) with consistency} \\ \hline
\bottomrule  
\end{tabular}
\end{minipage}

\begin{minipage}[t]{\textwidth}
\begin{tabular}{p{0.078in}|p{0.73in}|p{0.32in}|p{0.34in}|p{0.73in}|p{0.31in}}
\hline
\cellcolor{gray!30} B & \multicolumn{5}{c}{\textbf{Generator:} Read Skew Committed ($R_1[x_0]W_2[y_1]W_2[x_1]C_2R_1[y_1]$)}\\\hline
Q & \multicolumn{1}{c|}{Session 1: $T_1$-SQL} & \multicolumn{2}{c|}{Operations} & \multicolumn{1}{c|}{Session 2: $T_2$-SQL} & Result\\\hline
1 & Begin & & & & -\\\hline
\multirow{2}{*}{2} & SELECT * FROM t1  & \multirow{2}{*}{$R_1[x_0]$} \tikzmark{a2} & & & \multirow{2}{*}{(0,0)}\\
 &  WHERE k=0 &  & & & \\\hline
3 & & & & Begin & - \\\hline
\multirow{2}{*}{4} & & & \tikzmark{c2} \multirow{2}{*}{$W_2[y_1]$} & UPDATE t1 SET v=1 & \multirow{2}{*}{-} \\\
 & & & & WHERE k=1 &  \\\hline

\multirow{2}{*}{5} & & & \tikzmark{b2} \multirow{2}{*}{$W_2[x_1]$} & UPDATE t1 SET v=1  & \multirow{2}{*}{-} \\
 & & & & WHERE k=0 & \\\hline
6 & & & $C_2$ & \cellcolor{violet!20} Commit & -\\\hline

\multirow{2}{*}{7} & SELECT * FROM t1 & \multirow{2}{*}{$R_1[y_1]$} \tikzmark{d2} & & & \cellcolor{violet!20}MVCC+RC\\
 & WHERE k=1 & & & & \cellcolor{violet!20} (1,1) \\\hline
8 & Commit & $C_1$ & & & -\\\hline
\multicolumn{6}{c}{\cellcolor{yellow!30} \textbf{Checker:} Anomaly (\textbf{A}) detected} \\\hline
\bottomrule  

\end{tabular}
\end{minipage}

\begin{minipage}[t]{\textwidth}
\begin{tabular}{p{0.078in}|p{0.73in}|p{0.32in}|p{0.34in}|p{0.73in}|p{0.31in}}
\hline
\cellcolor{gray!30} C & \multicolumn{5}{c}{\textbf{Generator:} Lost Update Committed ($R_1[x_0]W_2[x_1]W_1C_2[x_2]$)}\\\hline
Q & \multicolumn{1}{c|}{Session 1: $T_1$-SQL} & \multicolumn{2}{c|}{Operations} & \multicolumn{1}{c|}{Session 2: $T_2$-SQL} & Result\\\hline
1 & Begin & & & & -\\\hline
\multirow{2}{*}{2} & SELECT * FROM t1  & \multirow{2}{*}{$R_1[x_0]$} \tikzmark{aa1}   & & & \multirow{2}{*}{(0,0)} \\
 &  WHERE k=0 &  & & & \\\hline
3 & & & & Begin & - \\\hline

\multirow{2}{*}{4} & & & \tikzmark{bb1}\multirow{2}{*}{$W_2[x_1]$}   & UPDATE t1 SET v=1  & \multirow{2}{*}{-} \\
 & & & & WHERE k=0 & \\\hline
 
5 & & & $C_2$ & \cellcolor{violet!20} Commit & -\\\hline

\multirow{2}{*}{6} &  UPDATE t1 SET v=1 & \multirow{2}{*}{$W_1[y_1]$} \tikzmark{cc1} & & & \cellcolor{violet!20}Abort by\\
 & WHERE k=0 &   & & &  \cellcolor{violet!20}rules\\\hline
\multicolumn{6}{c}{\cellcolor{cyan!30} \textbf{Checker:} Rollback (\textbf{R}) by rules (WCW)} \\\hline
\bottomrule  
\end{tabular}
\end{minipage}

\begin{minipage}[t]{\textwidth}
\begin{tabular}{p{0.078in}|p{0.75in}|p{0.76in}|p{0.75in}|p{0.31in}}
\hline

 \cellcolor{gray!30}                       & \multicolumn{4}{c}{  \textbf{Generator:} Step WAT} \\  
 \cellcolor{gray!30} \multirow{-2}{*}{  D }&   \multicolumn{4}{c}{  $W_1[x_1]W_2[y_1]W_3[z_1]W_3[y_2]W_2[x_2]W_1[z_2]$ }  \\ \hline

Q & \multicolumn{1}{c|}{Session 1: $T_1$-SQL} & \multicolumn{1}{c|}{Session 2: $T_2$-SQL} & \multicolumn{1}{c|}{Session 3: $T_3$-SQL} & Result\\\hline
1 & Begin & &  & -\\\hline

\multirow{2}{*}{2} &  $W_1[x_1]$: UPDATE t1  & Begin & & \multirow{2}{*}{-} \\
 &  SET v=1 WHERE k=0 \tikzmark{xx1}  &  & & \\\hline
 
\multirow{3}{*}{2} & & $W_2[y_2]$: UPDATE t1  & Begin & \multirow{2}{*}{-} \\\
 & & SET v=2 WHERE k=1 \tikzmark{yy1}  & &    \\\hline
 
\multirow{2}{*}{4} & &$W_2[x_2]$:  UPDATE t1  &  \tikzmark{zz1} $W_3[z_3]$: UPDATE t1  & \multirow{2}{*}{$W_2$ waited} \\
 & &  \tikzmark{xx2} SET v=2 WHERE k=0  &  SET v=3 WHERE k=2 &   \\\hline
 
 \multirow{2}{*}{5} & & & $W_3[y_3]$: UPDATE t1  & \multirow{2}{*}{$W_3$ waited} \\\
 & &  &  \tikzmark{yy2} SET v=3 WHERE k=1  &  \\\hline
 
\multirow{2}{*}{6} & $W_1[z_1]$: UPDATE t1  & & & \cellcolor{violet!20} 2PL Wait\\
 & SET v=1 WHERE k=2 \tikzmark{zz2} & & & \cellcolor{violet!20} deadlock \\\hline

\multicolumn{5}{c}{\cellcolor{red!30} \textbf{Checker:} Deadlock (\textbf{D}) detected} \\\hline
\bottomrule  
\end{tabular}
\end{minipage}

\begin{tikzpicture}[overlay, remember picture, shorten >=.5pt, shorten <=.5pt, transform canvas={yshift=.25\baselineskip}]
    \draw [magenta!70,->] ({pic cs:a1}) -- node [pos=0.45,above] {RW} ({pic cs:b1});
    \draw [magenta!70,->] ({pic cs:c1}) -- ({pic cs:d1})  node [pos=0.45,below] {RW};

    \draw [magenta!70,->] ({pic cs:a2}) -- node [pos=0.45,above] {RW} ({pic cs:b2});
    \draw [magenta!70,->] ({pic cs:c2}) -- ({pic cs:d2})  node [pos=0.45,below] {WCR};
    
        \draw [magenta!70,->] ({pic cs:aa1}) -- node [pos=0.45,above] {RW} ({pic cs:bb1});
    \draw [magenta!70,->] ({pic cs:bb1}) -- ({pic cs:cc1})  node [pos=0.45,below] {WCW};

    
    \draw [magenta!70,->] ({pic cs:xx1}) -- node [pos=0.45,above] {WW} ({pic cs:xx2});
    \draw [magenta!70,->] ({pic cs:yy1}) -- ({pic cs:yy2})  node [pos=0.45,below] {WW};
    \draw [magenta!70,->] ({pic cs:zz1}) -- node [pos=0.45,above] {WW} ({pic cs:zz2});
    
  \end{tikzpicture}
  \end{minipage}
}
\end{table}

\subsection{Construction of Test Cases}\label{sec:evaluation_test_case_construction}
We constructed all 33 types of data anomalies described in Table \ref{table:anomaly_classification}. Note that the SDA and DDA are finite and one-to-one mapping anomalies, yet MDA denotes a set of anomalies. So we design one typical case for each of the MDAs. Step RAT, Step WAT, and Step IAT have designed schedules with three WR, WW, and RW POPs, respectively.
For example, the schedule $S_2$ of the Read Skew anomaly can be executed in the following orders:
\begin{equation}
    S_2 = R_1[x_0]~\underline{W_2[x_1]~W_2[y_1]}~R_1[y_1]
\end{equation}
However, $W_2[y_1]$ may be waited as $W_2[x_1]$ may be waited by conflicting to $R_1[x_0]$, making the other conflict disappear. So, we may let non-conflict operations start first to simulate the complete conflicts. For example, the schedule $S_3$ of Read Skew anomaly can be executed in the following orders:
\begin{equation}
    S_3 = R_1[x_0]~\underline{W_2[y_1]~W_2[x_1]}~R_1[y_1]
\end{equation}
In the schedule $S_3$, note that $W_2[y_1]$ starts earlier than $W_2[x_1]$, as $W_2[y_1]$ does not conflict with $R_1[x_0]$. After the execution, we assure the occurrence of two conflicts by ($R_1[x_0]$, $W_2[x_1]$) and ($W_2[y_1]$, $R_1[y_1]$). $S_2$ and $S_3$ are actually equivalent schedule with the same version order $<_s$.
We then give the schedule of $S_4$ of Read Skew Committed anomaly in the following:
\begin{equation}
    S_4 = R_1[x_0]~W_2[y_1]~W_2[x_1]~C_2~R_1[y_1]
\end{equation}
As this time, the traditional conflict graph treated these $S_3$ and $S_4$ no difference. However, we recognized different POPs in $S_3$ and $S_4$, and later our evaluation will illustrate different performances under different isolation levels. Tables \ref{table:evaluation_test_case_read_skew}(A) and \ref{table:evaluation_test_case_read_skew}(B) depict the detailed preparation and execution steps by SQL queries for $S_3$ and $S_4$ by PostgreSQL at the RC level. In all our tests, the \textit{Begin} command is alone with the first operation while the \textit{Commit} command could be any order after schedule if not mentioned. PostgreSQL passed Read Skew schedule but found an anomaly result by Read Skew Committed schedule. Previous works (e.g., Elle \cite{DBLP:journals/pvldb/AlvaroK20}) often only detected anomaly cases, but in this paper, we also in-depth analyze the potential anomaly cases that are prevented by databases as shown in Table \ref{table:evaluation_test_case_read_skew}(A, C, and D).

For the construction of distributed databases,  we let keys (e.g., $k$=$0$ and $k$=$1$ in the Read Skew) spread into different distributed partitions (e.g., by PARTITION BY RANGE).  Greenplum, which by default has a write lock for a table/segment, needs to enable a global detector for supporting concurrent writes. Another way is to simulate the cases with multiple tables and each table having one row/key.


 \begin{table*}[hbt!]
	\caption{The consistency check results of 11 databases. Anomaly (\textcolor{yellow}{A}): Data anomalies are not recognized by the database, resulting in data inconsistencies. Consistency: The databases passed (\textcolor{green!50}{P}) the fed anomalies test case; or the database rollback a data anomaly by Rules (\textcolor{cyan!50}{R}), Deadlock Detection (\textcolor{red!50}{D}), or Timeout (\textcolor{blue!50}{T}) to guarantee consistency.}
	\label{table:evaluation_dbtest_result}
	\scriptsize
	\centering
	\setlength{\tabcolsep}{3.8pt}
\setlength\extrarowheight{1.1pt}

\begin{tabular}{|p{0.1in}|p{0.1in}|p{1.95in}|p{0.1in}
| p{0.1in} 
|> {\bfseries}p{0.1in}
|> {\bfseries} p{0.1in} |p{0.1in} |p{0.1in}   
|> {\bfseries} p{0.1in} |p{0.1in}
| p{0.1in} | p{0.1in}
|> {\bfseries} p{0.1in} |p{0.1in}| p{0.1in}
|> {\bfseries} p{0.1in}|p{0.1in}
| p{0.1in}| p{0.1in}| p{0.1in}
|> {\bfseries} p{0.1in}| p{0.1in}|p{0.1in}|p{0.1in}
|> {\bfseries} p{0.1in}|p{0.1in}|p{0.1in}|p{0.1in}|}
\toprule\hline

  \rotatebox[origin=c]{270}{The primary classification} &	 \multicolumn{2}{c|}{\rotatebox[origin=c]{270}{\begin{tabular}[c]{@{}c@{}}The data anomaly\\ instance test cases \\ from Table \ref{table:anomaly_classification}  \end{tabular}}} &	 \rotatebox[origin=c]{270}{Test case} &	 \rotatebox[origin=c]{270}{MongoDB 4.4.4} &	       \normalfont{\rotatebox[origin=c]{270}{CockroachDB v19.2.2}} &	 \multicolumn{3}{c|}{\rotatebox[origin=c]{270}{PostgreSQL 12.4}}  &				  \multicolumn{2}{c|}{ \rotatebox[origin=c]{270}{Greenplum 6.20.0}}  &				  \multicolumn{2}{c|}{ \rotatebox[origin=c]{270}{
  \begin{tabular}[c]{@{}c@{}} OceanBase 2.2.50 (MySQL)/ \\ OceanBase 3.1.2 (MySQL) \end{tabular}
  }}      &	  \multicolumn{3}{c|}{ \rotatebox[origin=c]{270}{
  \begin{tabular}[c]{@{}c@{}} OceanBase 2.2.50 (Oracle)/ \\ OceanBase 2.2.77 (Oracle) \end{tabular} }}       &			 \multicolumn{2}{c|}{\rotatebox[origin=c]{270}{\begin{tabular}[c]{@{}c@{}}  Oracle 12.1.0/ \\ Oracle 21.3.0\end{tabular} }} &		 \multicolumn{3}{c|}{\rotatebox[origin=c]{270}{
  \begin{tabular}[c]{@{}c@{}}  TiDB 4.0.5/ \\ TiDB 5.4.0 \end{tabular}
  }} &			 \multicolumn{4}{c|}{\rotatebox[origin=c]{270}{SQL Server 15.0}} &						 \multicolumn{4}{c|}{\rotatebox[origin=c]{270}{\begin{tabular}[c]{@{}c@{}} MySQL 8.0.20/ \\ MyRocks 8.0.26/ \\ TDSQL V2.0.1 \end{tabular}} }  \\ \hline			
\rotatebox[origin=c]{270}{Iso-level} &	 \multicolumn{2}{c|}{$\star$ for default levels}     &	 &	  \rotatebox[origin=c]{270}{SI} &	      \rotatebox[origin=c]{270}{ \textbf{SER}}   &	     \rotatebox[origin=c]{270}{ \textbf{SER}}    &	     \rotatebox[origin=c]{270}{RR}   &	   \rotatebox[origin=c]{270}{$\star$RC/RU}  &	    \rotatebox[origin=c]{270}{SER} &	    \rotatebox[origin=c]{270}{RC/RU}     &	    \rotatebox[origin=c]{270}{RR} &	     \rotatebox[origin=c]{270}{RC}    &	      \rotatebox[origin=c]{270}{ \textbf{SER}}        &	      \rotatebox[origin=c]{270}{RR}       &	     \rotatebox[origin=c]{270}{$\star$RC}    &	       \rotatebox[origin=c]{270}{ \textbf{SER}}     &	       \rotatebox[origin=c]{270}{$\star$RC}    &	        \rotatebox[origin=c]{270}{RR}      &	        \rotatebox[origin=c]{270}{RC}     &	        \rotatebox[origin=c]{270}{OPT}      &	    \rotatebox[origin=c]{270}{ \textbf{SER/RR}}  &	    \rotatebox[origin=c]{270}{SI} &	    \rotatebox[origin=c]{270}{$\star$RC/RCSI}  &	    \rotatebox[origin=c]{270}{RU} &	     \rotatebox[origin=c]{270}{ \textbf{SER}} &	    \rotatebox[origin=c]{270}{$\star$RR}  &	   \rotatebox[origin=c]{270}{RC}   &	   \rotatebox[origin=c]{270}{RU}        \\ \hline

   \multirow{15}{*}{\rotatebox[origin=c]{270}{RAT}}    																											
 &	 \multicolumn{2}{c|}{Dirty Read} &	 1  &	 \cellcolor{green!32} P &	 \cellcolor{green!32} P &	 \cellcolor{green!32} P &	 \cellcolor{green!32} P &	 \cellcolor{green!32} P &	 \cellcolor{green!32} P &	 \cellcolor{green!32} P &	 \cellcolor{green!32} P &	 \cellcolor{green!32} P &	 \cellcolor{green!32} P &	 \cellcolor{green!32} P &	 \cellcolor{green!32} P &	 \cellcolor{green!32} P &	 \cellcolor{green!32} P &	 \cellcolor{green!32} P &	 \cellcolor{green!32} P &	 \cellcolor{green!32} P &	 \cellcolor{green!32} P &	 \cellcolor{green!32} P &	 \cellcolor{green!32} P &	 \cellcolor{yellow!31} A &	 \cellcolor{green!32} P &	 \cellcolor{green!32} P &	 \cellcolor{green!32} P &	 \cellcolor{yellow!31} A       \\ \cline{2-29}
																											
 &	   \multicolumn{2}{c|}{Non-repeatable Read}&	  2   &	 \cellcolor{green!32} P &	 \cellcolor{green!32} P &	 \cellcolor{green!32} P &	 \cellcolor{green!32} P &	 \cellcolor{green!32} P &	 \cellcolor{green!32} P &	 \cellcolor{green!32} P &	 \cellcolor{green!32} P &	 \cellcolor{green!32} P &	 \cellcolor{green!32} P &	 \cellcolor{green!32} P &	 \cellcolor{green!32} P &	 \cellcolor{green!32} P &	 \cellcolor{green!32} P &	 \cellcolor{green!32} P &	 \cellcolor{green!32} P &	 \cellcolor{green!32} P &	 \cellcolor{green!32} P &	 \cellcolor{green!32} P &	 \cellcolor{yellow!31} A &	 \cellcolor{yellow!31} A &	 \cellcolor{green!32} P &	 \cellcolor{green!32} P &	 \cellcolor{green!32} P &	 \cellcolor{yellow!31} A      \\ \cline{2-29}

 &	  \multicolumn{2}{c|}{Intermediate Read} &	  3    &	 \cellcolor{green!32} P &	 \cellcolor{green!32} P &	 \cellcolor{green!32} P &	 \cellcolor{green!32} P &	 \cellcolor{green!32} P &	 \cellcolor{green!32} P &	 \cellcolor{green!32} P &	 \cellcolor{green!32} P &	 \cellcolor{green!32} P &	 \cellcolor{green!32} P &	 \cellcolor{green!32} P &	 \cellcolor{green!32} P &	 \cellcolor{green!32} P &	 \cellcolor{green!32} P &	 \cellcolor{green!32} P &	 \cellcolor{green!32} P &	 \cellcolor{green!32} P &	 \cellcolor{green!32} P &	 \cellcolor{green!32} P &	 \cellcolor{green!32} P &	 \cellcolor{yellow!31} A &	 \cellcolor{green!32} P &	 \cellcolor{green!32} P &	 \cellcolor{green!32} P &	 \cellcolor{yellow!31} A      \\ \cline{2-29}
 																											
 &	  \multicolumn{2}{c|}{Intermediate Read Committed} &	  4    &	 \cellcolor{green!32} P &	 \cellcolor{green!32} P &	 \cellcolor{green!32} P &	 \cellcolor{green!32} P &	 \cellcolor{green!32} P &	 \cellcolor{green!32} P &	 \cellcolor{green!32} P &	 \cellcolor{green!32} P &	 \cellcolor{green!32} P &	 \cellcolor{green!32} P &	 \cellcolor{green!32} P &	 \cellcolor{green!32} P &	 \cellcolor{green!32} P &	 \cellcolor{green!32} P &	 \cellcolor{green!32} P &	 \cellcolor{green!32} P &	 \cellcolor{green!32} P &	 \cellcolor{green!32} P &	 \cellcolor{green!32} P &	 \cellcolor{green!32} P &	 \cellcolor{yellow!31} A &	 \cellcolor{green!32} P &	 \cellcolor{green!32} P &	 \cellcolor{green!32} P &	 \cellcolor{yellow!31} A      \\ \cline{2-29}
  																											
&	 \multicolumn{2}{c|}{Lost self update} &	 5 &	 \cellcolor{green!32} P  &	 \cellcolor{green!32} P &	 \cellcolor{cyan!31} R &	 \cellcolor{cyan!31} R &	 \cellcolor{green!32} P &	 \cellcolor{cyan!31} R &	 \cellcolor{green!32} P &	 \cellcolor{cyan!31} R &	 \cellcolor{green!32} P &	 \cellcolor{cyan!31} R &	 \cellcolor{cyan!31} R &	 \cellcolor{green!32} P &	 \cellcolor{cyan!31} R &	 \cellcolor{green!32} P &	 \cellcolor{green!32} P &	 \cellcolor{green!32} P &	 \cellcolor{cyan!31} R &	 \cellcolor{green!32} P &	 \cellcolor{cyan!31} R &	 \cellcolor{green!32} P &	 \cellcolor{green!32} P &	 \cellcolor{green!32} P &	 \cellcolor{green!32} P &	 \cellcolor{green!32} P &	 \cellcolor{green!32} P        \\ \cline{2-29}

 &	 \multicolumn{2}{c|}{Write-read Skew} &	 6    &	 \cellcolor{yellow!31} A &	 \cellcolor{green!32} P &	 \cellcolor{cyan!31} R &	 \cellcolor{yellow!31} A &	 \cellcolor{yellow!31} A &	 \cellcolor{yellow!31} A &	 \cellcolor{yellow!31} A &	 \cellcolor{yellow!31} A &	 \cellcolor{yellow!31} A &	 \cellcolor{yellow!31} A &	 \cellcolor{yellow!31} A &	 \cellcolor{yellow!31} A &	 \cellcolor{yellow!31} A &	 \cellcolor{yellow!31} A &	 \cellcolor{yellow!31} A &	 \cellcolor{yellow!31} A &	 \cellcolor{yellow!31} A &	 \cellcolor{red!31} D &	 \cellcolor{yellow!31} A &	 \cellcolor{red!31} D &	 \cellcolor{yellow!31} A &	 \cellcolor{red!31} D &	 \cellcolor{yellow!31} A &	 \cellcolor{yellow!31} A &	 \cellcolor{yellow!31} A      \\ \cline{2-29}
																											
 &	 \multicolumn{2}{c|}{Write-read Skew Committed} &	    7    &	 \cellcolor{yellow!31} A &	 \cellcolor{green!32} P &	 \cellcolor{cyan!31} R &	 \cellcolor{yellow!31} A &	 \cellcolor{green!32} P &	 \cellcolor{yellow!31} A &	 \cellcolor{green!32} P &	 \cellcolor{yellow!31} A &	 \cellcolor{green!32} P &	 \cellcolor{yellow!31} A &	 \cellcolor{yellow!31} A &	 \cellcolor{green!32} P &	 \cellcolor{yellow!31} A &	 \cellcolor{green!32} P &	 \cellcolor{yellow!31} A &	 \cellcolor{green!32} P &	 \cellcolor{yellow!31} A &	 \cellcolor{red!31} D &	 \cellcolor{yellow!31} A &	 \cellcolor{red!31} D &	 \cellcolor{yellow!31} A &	 \cellcolor{red!31} D &	 \cellcolor{green!32} P &	 \cellcolor{green!32} P  &	 \cellcolor{yellow!31} A      \\ \cline{2-29}
																											
 &	  \multicolumn{2}{c|}{Double-write Skew 1}  &	 8 &	 \cellcolor{cyan!31} R &	 \cellcolor{cyan!31} R &	 \cellcolor{cyan!31} R &	 \cellcolor{cyan!31} R &	 \cellcolor{green!32} P &	 \cellcolor{cyan!31} R &	 \cellcolor{green!32} P &	 \cellcolor{cyan!31} R &	 \cellcolor{green!32} P &	 \cellcolor{cyan!31} R &	 \cellcolor{cyan!31} R &	 \cellcolor{green!32} P &	 \cellcolor{cyan!31} R &	 \cellcolor{green!32} P &	 \cellcolor{green!32} P &	 \cellcolor{green!32} P &	 \cellcolor{cyan!31} R &	 \cellcolor{red!31} D &	 \cellcolor{cyan!31} R &	 \cellcolor{red!31} D &	 \cellcolor{yellow!31} A &	 \cellcolor{red!31} D &	 \cellcolor{green!32} P &	 \cellcolor{green!32} P &	 \cellcolor{yellow!31} A          \\ \cline{2-29}
   																											
 &	 \multicolumn{2}{c|}{Double-write Skew 1 Committed} &	 9    &	 \cellcolor{cyan!31} R &	 \cellcolor{cyan!31} R &	 \cellcolor{cyan!31} R &	 \cellcolor{cyan!31} R &	 \cellcolor{green!32} P &	 \cellcolor{cyan!31} R &	 \cellcolor{green!32} P &	 \cellcolor{cyan!31} R &	 \cellcolor{green!32} P &	 \cellcolor{cyan!31} R &	 \cellcolor{cyan!31} R &	 \cellcolor{green!32} P &	 \cellcolor{cyan!31} R &	 \cellcolor{green!32} P &	 \cellcolor{green!32} P &	 \cellcolor{green!32} P &	 \cellcolor{cyan!31} R &	 \cellcolor{red!31} D &	 \cellcolor{cyan!31} R &	 \cellcolor{red!31} D &	 \cellcolor{yellow!31} A &	 \cellcolor{red!31} D &	 \cellcolor{green!32} P &	 \cellcolor{green!32} P &	 \cellcolor{yellow!31} A      \\ \cline{2-29}
																											
 &	  \multicolumn{2}{c|}{Double-write Skew 2}  &	 10 &	 \cellcolor{cyan!31} R &	 \cellcolor{green!32} P &	 \cellcolor{cyan!31} R &	 \cellcolor{cyan!31} R &	 \cellcolor{green!32} P &	 \cellcolor{cyan!31} R &	 \cellcolor{green!32} P &	 \cellcolor{cyan!31} R &	 \cellcolor{green!32} P &	 \cellcolor{cyan!31} R &	 \cellcolor{cyan!31} R &	 \cellcolor{green!32} P &	 \cellcolor{cyan!31} R &	 \cellcolor{green!32} P &	 \cellcolor{green!32} P &	 \cellcolor{green!32} P &	 \cellcolor{cyan!31} R &	 \cellcolor{red!31} D &	 \cellcolor{cyan!31} R &	 \cellcolor{red!31} D &	 \cellcolor{yellow!31} A &	 \cellcolor{red!31} D &	 \cellcolor{green!32} P &	 \cellcolor{green!32} P &	 \cellcolor{yellow!31} A           \\ \cline{2-29}
 																											
 &	 \multicolumn{2}{c|}{Read Skew} &	 11 &	 \cellcolor{green!32} P &	 \cellcolor{green!32} P &	 \cellcolor{green!32} P &	 \cellcolor{green!32} P &	 \cellcolor{green!32} P &	 \cellcolor{green!32} P &	 \cellcolor{green!32} P &	 \cellcolor{green!32} P &	 \cellcolor{green!32} P &	 \cellcolor{green!32} P &	 \cellcolor{green!32} P &	 \cellcolor{green!32} P &	 \cellcolor{green!32} P &	 \cellcolor{green!32} P &	 \cellcolor{green!32} P &	 \cellcolor{green!32} P &	 \cellcolor{green!32} P &	 \cellcolor{red!31} D &	 \cellcolor{green!32} P &	 \cellcolor{yellow!31} A &	 \cellcolor{yellow!31} A &	 \cellcolor{red!31} D &	 \cellcolor{green!32} P &	 \cellcolor{green!32} P &	 \cellcolor{yellow!31} A        \\ \cline{2-29}
 																											
 &	  \multicolumn{2}{c|}{Read Skew2}  &	 12    &	 \cellcolor{green!32} P &	 \cellcolor{green!32} P &	 \cellcolor{green!32} P &	 \cellcolor{green!32} P &	 \cellcolor{green!32} P &	 \cellcolor{green!32} P &	 \cellcolor{green!32} P &	 \cellcolor{green!32} P &	 \cellcolor{green!32} P &	 \cellcolor{green!32} P &	 \cellcolor{green!32} P &	 \cellcolor{green!32} P &	 \cellcolor{green!32} P &	 \cellcolor{green!32} P &	 \cellcolor{green!32} P &	 \cellcolor{green!32} P &	 \cellcolor{green!32} P &	 \cellcolor{red!31} D &	 \cellcolor{green!32} P &	 \cellcolor{yellow!31} A &	 \cellcolor{yellow!31} A &	 \cellcolor{red!31} D &	 \cellcolor{green!32} P &	 \cellcolor{green!32} P &	 \cellcolor{yellow!31} A                        \\ \cline{2-29}
 &	  \multicolumn{2}{c|}{Read Skew2 Committed}  &	 13   &	 \cellcolor{green!32} P &	 \cellcolor{green!32} P &	 \cellcolor{green!32} P &	 \cellcolor{green!32} P &	 \cellcolor{green!32} P &	 \cellcolor{green!32} P &	 \cellcolor{green!32} P &	 \cellcolor{green!32} P &	 \cellcolor{green!32} P &	 \cellcolor{green!32} P &	 \cellcolor{green!32} P &	 \cellcolor{green!32} P &	 \cellcolor{green!32} P &	 \cellcolor{green!32} P &	 \cellcolor{green!32} P &	 \cellcolor{green!32} P &	 \cellcolor{green!32} P &	 \cellcolor{red!31} D &	 \cellcolor{green!32} P &	 \cellcolor{yellow!31} A &	 \cellcolor{yellow!31} A &	 \cellcolor{red!31} D &	 \cellcolor{green!32} P &	 \cellcolor{green!32} P &	 \cellcolor{yellow!31} A                        \\ \cline{2-29}

 &	 \multicolumn{2}{c|}{Step RAT}  &	 14 &	 \cellcolor{yellow!31} A &	 \cellcolor{green!32} P &	 \cellcolor{cyan!31} R &	 \cellcolor{yellow!31} A &	 \cellcolor{yellow!31} A &	 \cellcolor{yellow!31} A &	 \cellcolor{yellow!31} A &	 \cellcolor{yellow!31} A &	 \cellcolor{yellow!31} A &	 \cellcolor{yellow!31} A &	 \cellcolor{yellow!31} A &	 \cellcolor{yellow!31} A &	 \cellcolor{yellow!31} A &	 \cellcolor{yellow!31} A &	 \cellcolor{yellow!31} A &	 \cellcolor{yellow!31} A &	 \cellcolor{yellow!31} A &	 \cellcolor{red!31} D &	 \cellcolor{yellow!31} A &	 \cellcolor{red!31} D &	 \cellcolor{yellow!31} A &	 \cellcolor{red!31} D &	 \cellcolor{yellow!31} A &	 \cellcolor{yellow!31} A &	 \cellcolor{yellow!31} A   \\  \hline 

 \multirow{11}{*}{\rotatebox[origin=c]{270}{WAT}}   																											
 &	 \multicolumn{2}{c|}{Dirty Write}  &	 15 &	 \cellcolor{cyan!31} R &	 \cellcolor{green!32} P &	 \cellcolor{cyan!31} R &	 \cellcolor{cyan!31} R &	 \cellcolor{green!32} P &	 \cellcolor{cyan!32} R &	 \cellcolor{green!32} P &	 \cellcolor{cyan!32} R &	 \cellcolor{green!32} P &	 \cellcolor{cyan!32} R &	 \cellcolor{cyan!32} R &	 \cellcolor{green!32} P &	 \cellcolor{cyan!32} R &	 \cellcolor{green!32} P &	 \cellcolor{green!32} P &	 \cellcolor{green!32} P &	 \cellcolor{cyan!32} R &	 \cellcolor{green!32} P &	 \cellcolor{cyan!32} R &	 \cellcolor{green!32} P &	 \cellcolor{green!32} P &	 \cellcolor{green!32} P &	 \cellcolor{green!32} P &	 \cellcolor{green!32} P &	 \cellcolor{green!32} P           \\ \cline{2-29}
  																											
 &	 \multicolumn{2}{c|}{Full Write} &	 16  &	 \cellcolor{cyan!31} R &	 \cellcolor{green!32} P &	 \cellcolor{cyan!31} R &	 \cellcolor{cyan!31} R &	 \cellcolor{green!32} P &	 \cellcolor{cyan!31} R &	 \cellcolor{green!32} P &	 \cellcolor{cyan!31} R &	 \cellcolor{green!32} P &	 \cellcolor{cyan!31} R &	 \cellcolor{cyan!31} R &	 \cellcolor{green!32} P &	 \cellcolor{cyan!31} R &	 \cellcolor{green!32} P &	 \cellcolor{green!32} P &	 \cellcolor{green!32} P &	 \cellcolor{cyan!31} R &	 \cellcolor{green!32} P &	 \cellcolor{cyan!31} R &	 \cellcolor{green!32} P &	 \cellcolor{green!32} P &	 \cellcolor{green!32} P &	 \cellcolor{green!32} P &	 \cellcolor{green!32} P &	 \cellcolor{green!32} P         \\ \cline{2-29}
 																											
 &	 \multicolumn{2}{c|}{Full Write Committed} &	 17  &	 \cellcolor{cyan!31} R &	 \cellcolor{green!32} P &	 \cellcolor{cyan!31} R &	 \cellcolor{cyan!31} R &	 \cellcolor{green!32} P &	 \cellcolor{cyan!31} R &	 \cellcolor{green!32} P &	 \cellcolor{cyan!31} R &	 \cellcolor{green!32} P &	 \cellcolor{cyan!31} R &	 \cellcolor{cyan!31} R &	 \cellcolor{green!32} P &	 \cellcolor{cyan!31} R &	 \cellcolor{green!32} P &	 \cellcolor{green!32} P &	 \cellcolor{green!32} P &	 \cellcolor{cyan!31} R &	 \cellcolor{green!32} P &	 \cellcolor{cyan!31} R &	 \cellcolor{green!32} P &	 \cellcolor{green!32} P &	 \cellcolor{green!32} P &	 \cellcolor{green!32} P &	 \cellcolor{green!32} P &	 \cellcolor{green!32} P       \\ \cline{2-29}
 																											
 &	  \multicolumn{2}{c|}{Lost Update} &	 18  &	 \cellcolor{cyan!31} R &	 \cellcolor{cyan!31} R &	 \cellcolor{cyan!31} R &	 \cellcolor{cyan!31} R &	 \cellcolor{yellow!31} A &	 \cellcolor{cyan!31} R &	 \cellcolor{yellow!31} A &	 \cellcolor{cyan!31} R &	 \cellcolor{yellow!31} A &	 \cellcolor{cyan!31} R &	 \cellcolor{cyan!31} R &	 \cellcolor{yellow!31} A &	 \cellcolor{cyan!31} R &	 \cellcolor{yellow!31} A &	 \cellcolor{yellow!31} A &	 \cellcolor{yellow!31} A &	 \cellcolor{cyan!31} R &	 \cellcolor{green!32} P &	 \cellcolor{cyan!31} R &	 \cellcolor{yellow!31} A &	 \cellcolor{yellow!31} A &	 \cellcolor{red!31} D &	 \cellcolor{yellow!31} A &	 \cellcolor{yellow!31} A  &	 \cellcolor{yellow!31} A           \\ \cline{2-29}
  																											
 &	 \multicolumn{2}{c|}{Lost Self Update Committed} &	 19  &	 \cellcolor{cyan!31} R &	 \cellcolor{green!32} P &	 \cellcolor{cyan!31} R &	 \cellcolor{cyan!31} R &	 \cellcolor{green!32} P &	 \cellcolor{cyan!31} R &	 \cellcolor{green!32} P &	 \cellcolor{cyan!31} R &	 \cellcolor{green!32} P &	 \cellcolor{cyan!31} R &	 \cellcolor{cyan!31} R &	 \cellcolor{green!32} P &	 \cellcolor{cyan!31} R &	 \cellcolor{green!32} P &	 \cellcolor{green!32} P &	 \cellcolor{green!32} P &	 \cellcolor{cyan!31} R &	 \cellcolor{green!32} P &	 \cellcolor{cyan!31} R &	 \cellcolor{green!32} P &	 \cellcolor{green!32} P &	 \cellcolor{green!32} P &	 \cellcolor{green!32} P &	 \cellcolor{green!32} P  &	 \cellcolor{green!32} P         \\ \cline{2-29}
   																											
  &	  \multicolumn{2}{c|}{Double-write Skew 2 Committed}  &	 20 &	 \cellcolor{cyan!31} R &	 \cellcolor{green!32} P &	 \cellcolor{cyan!31} R &	 \cellcolor{cyan!31} R &	 \cellcolor{green!32} P &	 \cellcolor{cyan!31} R &	 \cellcolor{green!32} P &	 \cellcolor{cyan!31} R &	 \cellcolor{green!32} P &	 \cellcolor{cyan!31} R &	 \cellcolor{cyan!31} R &	 \cellcolor{green!32} P &	 \cellcolor{cyan!31} R &	 \cellcolor{green!32} P &	 \cellcolor{green!32} P &	 \cellcolor{green!32} P &	 \cellcolor{cyan!31} R &	 \cellcolor{red!31} D &	 \cellcolor{cyan!31} R &	 \cellcolor{red!31} D &	 \cellcolor{yellow!31} A &	 \cellcolor{red!31} D &	 \cellcolor{green!32} P &	 \cellcolor{green!32} P &	 \cellcolor{yellow!31} A         \\ \cline{2-29}
 																											
 &	 \multicolumn{2}{c|}{Full-write Skew} &	 21  &	 \cellcolor{cyan!31} R &	 \cellcolor{red!31} D &	 \cellcolor{red!31} D &	 \cellcolor{red!31} D &	 \cellcolor{red!31} D &	 \cellcolor{red!31} D &	 \cellcolor{red!31} D &	 \cellcolor{blue!31} T &	 \cellcolor{blue!31} T &	 \cellcolor{blue!31} T &	 \cellcolor{blue!31} T &	 \cellcolor{blue!31} T &	 \cellcolor{red!31} D &	 \cellcolor{red!31} D &	 \cellcolor{red!31} D &	 \cellcolor{red!31} D &	 \cellcolor{cyan!31} R &	 \cellcolor{red!31} D &	 \cellcolor{red!31} D &	 \cellcolor{red!31} D &	 \cellcolor{red!31} D &	 \cellcolor{red!31} D &	 \cellcolor{red!31} D &	 \cellcolor{red!31} D &	 \cellcolor{red!31} D        \\ \cline{2-29}
 																											
 &	 \multicolumn{2}{c|}{Full-write Skew Committed} &	 22 &	 \cellcolor{cyan!31} R &	 \cellcolor{red!31} D &	 \cellcolor{red!31} D &	 \cellcolor{red!31} D &	 \cellcolor{red!31} D &	 \cellcolor{red!31} D &	 \cellcolor{red!31} D &	 \cellcolor{blue!31} T &	 \cellcolor{blue!31} T &	 \cellcolor{blue!31} T &	 \cellcolor{blue!31} T &	 \cellcolor{blue!31} T &	 \cellcolor{red!31} D &	 \cellcolor{red!31} D &	 \cellcolor{red!31} D &	 \cellcolor{red!31} D &	 \cellcolor{cyan!31} R &	 \cellcolor{red!31} D &	 \cellcolor{red!31} D &	 \cellcolor{red!31} D &	 \cellcolor{red!31} D &	 \cellcolor{red!31} D &	 \cellcolor{red!31} D &	 \cellcolor{red!31} D &	 \cellcolor{red!31} D        \\ \cline{2-29}
																											
  &	 \multicolumn{2}{c|}{Read-write Skew 1} &	 23  &	 \cellcolor{cyan!31} R &	 \cellcolor{cyan!31} R &	 \cellcolor{cyan!31} R &	 \cellcolor{cyan!31} R &	 \cellcolor{yellow!31} A &	 \cellcolor{cyan!31} R &	 \cellcolor{yellow!31} A &	 \cellcolor{cyan!31} R &	 \cellcolor{yellow!31} A &	 \cellcolor{cyan!31} R &	 \cellcolor{cyan!31} R &	 \cellcolor{yellow!31} A &	 \cellcolor{cyan!31} R &	 \cellcolor{yellow!31} A &	 \cellcolor{yellow!31} A &	 \cellcolor{yellow!31} A &	 \cellcolor{cyan!31} R &	 \cellcolor{red!31} D &	 \cellcolor{cyan!31} R &	 \cellcolor{yellow!31} A &	 \cellcolor{yellow!31} A &	 \cellcolor{red!31} D &	 \cellcolor{yellow!31} A &	 \cellcolor{yellow!31} A &	 \cellcolor{yellow!31} A          \\ \cline{2-29}
																											
  &	 \multicolumn{2}{c|}{Read-write Skew 2} &	 24  &	 \cellcolor{cyan!31} R &	 \cellcolor{cyan!31} R &	 \cellcolor{cyan!31} R &	 \cellcolor{cyan!31} R &	 \cellcolor{yellow!31} A &	 \cellcolor{cyan!31} R &	 \cellcolor{yellow!31} A &	 \cellcolor{cyan!31} R &	 \cellcolor{yellow!31} A &	 \cellcolor{cyan!31} R &	 \cellcolor{cyan!31} R &	 \cellcolor{yellow!31} A &	 \cellcolor{cyan!31} R &	 \cellcolor{yellow!31} A &	 \cellcolor{yellow!31} A &	 \cellcolor{yellow!31} A &	 \cellcolor{cyan!31} R &	 \cellcolor{red!31} D &	 \cellcolor{cyan!31} R &	 \cellcolor{yellow!31} A &	 \cellcolor{yellow!31} A &	 \cellcolor{red!31} D &	 \cellcolor{yellow!31} A &	 \cellcolor{yellow!31} A &	 \cellcolor{yellow!31} A          \\ \cline{2-29}  
																											
  &	 \multicolumn{2}{c|}{Read-write Skew 2 Committed} &	 25  &	 \cellcolor{cyan!31} R &	 \cellcolor{cyan!31} R &	 \cellcolor{cyan!31} R &	 \cellcolor{cyan!31} R &	 \cellcolor{yellow!31} A &	 \cellcolor{cyan!31} R &	 \cellcolor{yellow!31} A &	 \cellcolor{cyan!31} R &	 \cellcolor{yellow!31} A &	 \cellcolor{cyan!31} R &	 \cellcolor{cyan!31} R &	 \cellcolor{yellow!31} A &	 \cellcolor{cyan!31} R &	 \cellcolor{yellow!31} A &	 \cellcolor{yellow!31} A &	 \cellcolor{yellow!31} A &	 \cellcolor{cyan!31} R &	 \cellcolor{red!31} D &	 \cellcolor{cyan!31} R &	 \cellcolor{yellow!31} A &	 \cellcolor{yellow!31} A &	 \cellcolor{red!31} D &	 \cellcolor{yellow!31} A &	 \cellcolor{yellow!31} A &	 \cellcolor{yellow!31} A       \\ \cline{2-29}  
 																											
 &	  \multicolumn{2}{c|}{Step WAT} &	 26 &	 \cellcolor{cyan!31} R &	 \cellcolor{red!31} D &	 \cellcolor{red!31} D &	 \cellcolor{red!31} D &	 \cellcolor{red!31} D &	 \cellcolor{red!31} D &	 \cellcolor{green!32} P &	 \cellcolor{blue!31} T &	 \cellcolor{blue!31} T &	 \cellcolor{blue!31} T &	 \cellcolor{blue!31} T &	 \cellcolor{blue!31} T &	 \cellcolor{red!31} D &	 \cellcolor{red!31} D &	 \cellcolor{red!31} D &	 \cellcolor{red!31} D &	 \cellcolor{cyan!31} R &	 \cellcolor{red!31} D &	 \cellcolor{red!31} D &	 \cellcolor{red!31} D &	 \cellcolor{red!31} D &	 \cellcolor{red!31} D &	 \cellcolor{red!31} D &	 \cellcolor{red!31} D &	 \cellcolor{red!31} D     \\ \hline

 \multirow{8}{*}{\rotatebox[origin=c]{270}{IAT}}   																											
 &	 \multicolumn{2}{c|}{Non-repeatable Read Committed} &	  27    &	 \cellcolor{green!32} P &	 \cellcolor{green!32} P &	 \cellcolor{green!32} P &	 \cellcolor{green!32} P &	 \cellcolor{yellow!31} A &	 \cellcolor{green!32} P &	 \cellcolor{yellow!31} A &	 \cellcolor{green!32} P &	 \cellcolor{yellow!31} A &	 \cellcolor{green!32} P &	 \cellcolor{green!32} P &	 \cellcolor{yellow!31} A &	 \cellcolor{green!32} P &	 \cellcolor{yellow!31} A &	 \cellcolor{green!32} P &	 \cellcolor{yellow!31} A &	 \cellcolor{green!32} P &	 \cellcolor{green!32} P &	 \cellcolor{green!32} P &	 \cellcolor{yellow!31} A &	 \cellcolor{yellow!31} A &	 \cellcolor{green!32} P &	 \cellcolor{green!32} P &	 \cellcolor{yellow!31} A &	 \cellcolor{yellow!31} A     \\ \cline{2-29}
																											
  &	 \multicolumn{2}{c|}{Lost Update Committed} &	 28 &	 \cellcolor{cyan!31} R &	 \cellcolor{cyan!31} R &	 \cellcolor{cyan!31} R &	 \cellcolor{cyan!31} R &	 \cellcolor{yellow!31} A &	 \cellcolor{cyan!31} R &	 \cellcolor{yellow!31} A &	 \cellcolor{cyan!31} R &	 \cellcolor{yellow!31} A &	 \cellcolor{cyan!31} R &	 \cellcolor{cyan!31} R &	 \cellcolor{yellow!31} A &	 \cellcolor{cyan!31} R &	 \cellcolor{yellow!31} A &	 \cellcolor{yellow!31} A &	 \cellcolor{yellow!31} A &	 \cellcolor{cyan!31} R &	 \cellcolor{green!32} P &	 \cellcolor{cyan!31} R &	 \cellcolor{yellow!31} A &	 \cellcolor{yellow!31} A &	 \cellcolor{red!31} D &	 \cellcolor{yellow!31} A &	 \cellcolor{yellow!31} A &	 \cellcolor{yellow!31} A      \\ \cline{2-29}
																											
 &	   \multicolumn{2}{c|}{Read Skew Committed}  &	 29 &	 \cellcolor{green!32} P &	 \cellcolor{green!32} P &	 \cellcolor{green!32} P &	 \cellcolor{green!32} P &	 \cellcolor{yellow!31} A &	 \cellcolor{green!32} P &	 \cellcolor{yellow!31} A &	 \cellcolor{green!32} P &	 \cellcolor{yellow!31} A &	 \cellcolor{green!32} P &	 \cellcolor{green!32} P &	 \cellcolor{yellow!31} A &	 \cellcolor{green!32} P &	 \cellcolor{yellow!31} A &	 \cellcolor{green!32} P &	 \cellcolor{yellow!31} A &	 \cellcolor{green!32} P &	 \cellcolor{red!31} D &	 \cellcolor{green!32} P &	 \cellcolor{yellow!31} A &	 \cellcolor{yellow!31} A &	 \cellcolor{red!31} D &	 \cellcolor{green!32} P &	 \cellcolor{yellow!31} A &	 \cellcolor{yellow!31} A      \\ \cline{2-29}
																											
 &	  \multicolumn{2}{c|}{Read-write Skew 1 Committed} &	  30    &	 \cellcolor{cyan!31} R &	 \cellcolor{cyan!31} R &	 \cellcolor{cyan!31} R &	 \cellcolor{cyan!31} R &	 \cellcolor{yellow!31} A &	 \cellcolor{cyan!31} R &	 \cellcolor{yellow!31} A  &	 \cellcolor{cyan!31} R &	 \cellcolor{yellow!31} A  &	 \cellcolor{cyan!31} R &	 \cellcolor{cyan!31} R &	 \cellcolor{yellow!31} A  &	 \cellcolor{cyan!31} R &	 \cellcolor{yellow!31} A &	 \cellcolor{yellow!31} A &	 \cellcolor{yellow!31} A &	 \cellcolor{cyan!31} R &	 \cellcolor{red!31} D &	 \cellcolor{cyan!31} R &	 \cellcolor{yellow!31} A &	 \cellcolor{yellow!31} A &	 \cellcolor{red!31} D &	 \cellcolor{yellow!31} A &	 \cellcolor{yellow!31} A &	 \cellcolor{yellow!31} A      \\ \cline{2-29}
																											
 &	  \multicolumn{2}{c|}{Write Skew}  &	 31    &	 \cellcolor{yellow!31} A &	 \cellcolor{cyan!31} R &	 \cellcolor{cyan!31} R &	 \cellcolor{yellow!31} A &	 \cellcolor{yellow!31} A &	 \cellcolor{yellow!31} A &	 \cellcolor{yellow!31} A &	 \cellcolor{yellow!31} A &	 \cellcolor{yellow!31} A &	 \cellcolor{yellow!31} A &	 \cellcolor{yellow!31} A &	 \cellcolor{yellow!31} A &	 \cellcolor{yellow!31} A &	 \cellcolor{yellow!31} A &	 \cellcolor{yellow!31} A &	 \cellcolor{yellow!31} A &	 \cellcolor{yellow!31} A &	 \cellcolor{red!31} D &	 \cellcolor{yellow!31} A &	 \cellcolor{yellow!31} A &	 \cellcolor{yellow!31} A &	 \cellcolor{red!31} D &	 \cellcolor{yellow!31} A &	 \cellcolor{yellow!31} A &	 \cellcolor{yellow!31} A     \\ \cline{2-29}
  &	  \multicolumn{2}{c|}{Write Skew Committed}  &	 32    &	 \cellcolor{yellow!31} A &	 \cellcolor{cyan!31} R &	 \cellcolor{cyan!31} R &	 \cellcolor{yellow!31} A &	 \cellcolor{yellow!31} A &	 \cellcolor{yellow!31} A &	 \cellcolor{yellow!31} A &	 \cellcolor{yellow!31} A &	 \cellcolor{yellow!31} A &	 \cellcolor{yellow!31} A &	 \cellcolor{yellow!31} A &	 \cellcolor{yellow!31} A &	 \cellcolor{cyan!31} R &	 \cellcolor{yellow!31} A &	 \cellcolor{yellow!31} A &	 \cellcolor{yellow!31} A &	 \cellcolor{yellow!31} A &	 \cellcolor{red!31} D &	 \cellcolor{yellow!31} A &	 \cellcolor{yellow!31} A &	 \cellcolor{yellow!31} A &	 \cellcolor{red!31} D &	 \cellcolor{yellow!31} A &	 \cellcolor{yellow!31} A &	 \cellcolor{yellow!31} A     \\ \cline{2-29}
																											
 &	  \multicolumn{2}{c|}{Step IAT} &	 33   &	 \cellcolor{yellow!31} A &	 \cellcolor{cyan!31} R &	 \cellcolor{cyan!31} R &	 \cellcolor{yellow!31} A &	 \cellcolor{yellow!31} A &	 \cellcolor{yellow!31} A &	 \cellcolor{yellow!31} A &	 \cellcolor{yellow!31} A &	 \cellcolor{yellow!31} A &	 \cellcolor{yellow!31} A &	 \cellcolor{yellow!31} A &	 \cellcolor{yellow!31} A &	 \cellcolor{cyan!31} R &	 \cellcolor{yellow!31} A &	 \cellcolor{yellow!31} A &	 \cellcolor{yellow!31} A &	 \cellcolor{yellow!31} A &	 \cellcolor{red!31} D &	 \cellcolor{yellow!31} A &	 \cellcolor{yellow!31} A &	 \cellcolor{yellow!31} A &	 \cellcolor{red!31} D &	 \cellcolor{yellow!31} A &	 \cellcolor{yellow!31} A &	 \cellcolor{yellow!31} A                        \\ \hline

\bottomrule
\end{tabular}

\end{table*}

\subsection{Consistency Check in Databases}\label{sec:evaluation_consistency_check}
This part provides a general summary of the evaluation results. Table \ref{table:evaluation_dbtest_result} shows the overall evaluation result of 11 databases with different isolation levels by 33 test cases constructed via SQL queries (except for MongoDB). The evaluation is cost-effective and reproducible, as we do not rely on the time- and resource-consuming random workloads but specifically and determinedly generate representative inconsistent scenarios. The average time spent for each level to finish 33 tests is around 1 minute. The original and executed schedules are available for analysis and debugging. The result behaviors are classified into two types, i.e., anomaly (A) and consistency. For anomaly occurrence, data anomalies are not recognized by databases, resulting in data inconsistencies, meaning the executed schedule with no equivalent serializable execution (or a POP cycle). While for the consistent performance, databases either pass (P) the anomaly test cases with a serializable result (no POP cycle) cycle or rollback transactions due to rules (R), deadlock detection (D), or timeout (T) reached.


\noindent \textbf{SER level:} \indent All tested databases can guarantee no anomalies except Oracle, OceanBase, and Greenplum. These three databases claimed SER levels yet performed equivalent to the SI level. 
As most knowledge, researchers discover Oracle's inconsistency at its SER level by the Write Skew anomaly \cite{DBLP:conf/sigmod/BerensonBGMOO95}. 
However, we found that anomalies also happened when feeding test cases of Write-read Skew, Write-read Skew Committed, Write Skew Committed, Step RAT, and Step IAT. These are similar anomalies yet previous work is hard to quantify such cases. More importantly, by Coo, we can build an infinite number of various-object Step RAT and Step IAT to reproduce anomaly scenarios, which are non-trivial by traditional tests or by CC protocols. 

\noindent \textbf{Weaker isolation levels:} \indent Unlike the original isolation levels having a coarse sense of a few anomalies, we can recognize and analyze many more newly found anomalies between levels, and some anomalies are confused to fit into one specific level (e.g., Lost Update Committed aborted in PostgreSQL as shown in Table \ref{table:evaluation_test_case_read_skew}(C) and appeared in MySQL at the RR level). 
More POPs are allowed at weaker levels and some anomalies are expected to be appeared by combinations of these allowed POPs. Roughly speaking, anomalies of RAT types have most Pass cases. In contrast, anomalies of WAT types have the most different rollback cases while databases occur most anomalies with test cases of IAT types. We explain more details of POP behaviors and anomaly occurrences in the right following.

\subsection{Detailed Evaluation of POP Graphs}\label{sec:evaluation_weaker_isolation}
This part explains more details of POP behaviors and data anomaly occurrences. Specifically, we discuss consistency or consistent behaviors via POP and POP cycles. Firstly, POPs are the unit of conflicts that are handled by CC protocols (e.g., MVCC \cite{DBLP:journals/tods/BernsteinG83} and 2PL \cite{DBLP:journals/cacm/EswarranGLT76}). CC protocols perform different rules to allow or forbid these POPs. Roughly speaking, MySQL/RocksDB and TiDB are mainly using 2PL, and support MVCC at RR and RC levels. SQL Server uses pure 2PL and supports MVCC at its SI level. Other databases support MVCC at all levels and use 2PL for write locks.
Secondly, POP cycles are specific anomalies, and consistency is guaranteed if cycles are destroyed based on these POP behaviors. 

\subsubsection{\textbf{POP Behaviors}} \label{sec:POPBehaviors}
\
\newline
This part discusses the behaviors of POPs at different isolation levels. Table \ref{table:evaluation_rollback_strategies} shows a summary of behaviors of three primitive POPs, i.e., WR, WW, and RW, corresponding to our test cases in three types, i.e., RAT, WAT, IAT. Core CC protocols used in different databases are 2PL and MVCC. Some are using combined protocols. 


The \textbf{WR} is waited by MySQL at the SER level or SQL Server at SER, RR, and RC levels, and is allowed in other cases. First, WR is indeed allowed in 2PL databases at the RU level, as they allow a read on an uncommitted write.  Second,  WR is allowed by MVCC (e.g., PostgreSQL at all levels and MySQL at RR and RC levels) by reading the old committed version, transforming it into RW. 
For example, MySQL executed the Intermediate Reads ($W_1[x_1]$ $R_2[x_1]$ $W_1[x_2]$) as expected at the RU level but into a non-anomaly ($W_1[x_1]$ $R_2[x_0]$ $W_1[x_2]$) at the RC level.
      
The \textbf{WW} is waited by most evaluated databases at any level, except MongoDB directly aborts it and TiDB, at OPT level, prewrites it in private. 
This is very different from the ANIS SQL standard that considers WW as a Dirty write and forbids it at any level. In practice, the WW is somehow waited (not immediate abort) by the 2PL Wait strategy.
For example, MySQL and SQL Server passed Full Write anomalies ($W_1[x_1]W_2[x_2]W_1[x_3]$), as they executed it into a non-anomaly ($W_1[x_1]W_1[x_3]C_1W_2[x_2]$), transforming WW into WCW (will discuss later).
    
The \textbf{RW} is allowed by most evaluated databases at any level, except at the SER level, 2PL databases wait for it and CockroachDB aborts it. Note SQL Server by using 2PL at the RR level still waits for RW. 
For example when executing Write Skew anomaly ($R_1[x_0]$ $R_2[y_0]W_2[x_1]W_1[y_1]$) at the SER level, 2PL databases (e.g., SQL Server) waited for each other by two RWs, i.e., ($R_1[x_0]W_2[x_1]$ and $R_2[y_0]W_1[y_1]$), yielding deadlocks.
PostgreSQL allowed each RW in Write Skew but aborted it when two consecutive RWs were formed by the SSI \cite{DBLP:journals/pvldb/PortsG12} at the SER level while passed it as an anomaly at other levels.

\begin{table}[t!]
	\caption{Databases behaviors when meeting WW, WR, and RW POPs. 2PL(wait)/2PL(abort) stands for the waiting/abort of POPs, and MV(trans) stands for the transformation from WR to RW by MVCC.}
	\label{table:evaluation_rollback_strategies}
	\tiny
	\footnotesize
	\centering




\setlength{\tabcolsep}{3pt}
\begin{tabular}{cccccc}
\toprule
\hline
POPs                 & DBs       & SER              & RR                & RC                & RU                \\\hline

\multirow{11}{*}{WR} 
                    & MySQL/TDSQL     & 2PL(wait)              & MV(trans)             & MV(trans)             & allow            \\\cline{2-6}
                    & SQL Server  & 2PL(wait)               & 2PL(wait)               & 2PL(wait)             & allow   
                    \\\cline{2-6}
                    & SQL Server (SI)  & /              & MV(trans)              & MV(trans)             & /   
                    \\\cline{2-6}
                    & TiDB     & /                 & MV(trans)              & MV(trans)              & /                 \\\cline{2-6}
                    & TiDB (OPT)       & /                 & /              & MV(trans)             & /                 \\\cline{2-6}
                    & Oracle     & MV(trans) & /                 & MV(trans) & /                 \\\cline{2-6}
                    & OceanBase (Oracle)   & MV(trans) & MV(trans) & MV(trans) & /                 \\\cline{2-6}
                     & OceanBase (MySQL)   & / & MV(trans) & MV(trans) & /                 \\\cline{2-6}
                    & Greenplum   & MV(trans) & / & MV(trans) & MV(trans)                 \\\cline{2-6}
                    & PostgreSQL          & MV(trans) & MV(trans) & MV(trans) & MV(trans) \\\cline{2-6}
                    & CockroachDB         & MV(trans) & /                 & /                 & /                 \\\cline{2-6}
                    & MongoDB     & /                 & MV(trans) & /                 & /                 \\\hline

\multirow{10}{*}{WW} 
                    & MySQL/TDSQL       &2PL(wait)              & 2PL(wait)              & 2PL(wait)              & 2PL(wait)              \\\cline{2-6}
                    & SQL Server  & 2PL(wait)              & 2PL(wait)              & 2PL(wait)              & 2PL(wait) 
                    \\\cline{2-6}
                    & SQL Server (SI)  & /              & 2PL(wait)              & 2PL(wait)              & / 
                    \\\cline{2-6}
                    & TiDB        & /                 & 2PL(wait)              & 2PL(wait)              & /                 \\\cline{2-6}
                    & TiDB (OPT)       & /                 & /              & prewrite              & /                 \\\cline{2-6}
                    & Oracle      & 2PL(wait)              & /                 & 2PL(wait)              & /                 \\\cline{2-6}
                    & OceanBase (Oracle)   & 2PL(wait)              & 2PL(wait)              & 2PL(wait)              & /                 \\\cline{2-6}
                    & OceanBase (MySQL)   & /              & 2PL(wait)              & 2PL(wait)              & /                 \\\cline{2-6}
                    & Greenplum   & 2PL(wait) & / & 2PL(wait) & 2PL(wait)                 \\\cline{2-6}
                    & PostgreSQL         & 2PL(wait)              & 2PL(wait)              & 2PL(wait)              & 2PL(wait)              \\\cline{2-6}
                    
                    & CockroachDB         & 2PL(wait)              & /                 & /                 & /                 \\\cline{2-6}
                    & MongoDB     & /                 & 2PL(abort)             & /                 & /                 \\\hline

\multirow{10}{*}{RW} 
                    & MySQL/TDSQL      & 2PL(wait)              & allow            & allow             & allow             \\\cline{2-6}
                    & SQL Server     & 2PL(wait)              & 2PL(wait)              & allow             & allow             \\\cline{2-6}
                    & SQL Server (SI)  & /              & allow              & allow             & /             \\\cline{2-6}
                    & TiDB        & /                 & allow             & allow             & /                 \\\cline{2-6}         
                    & TiDB (OPT)       & /                 & allow             & allow             & /                 \\\cline{2-6}

                    & Oracle     & allow             & /                 & allow             & /                 \\\cline{2-6}
                    
                    & OceanBase (Oracle)   & allow             & allow             & allow             & /                 \\\cline{2-6}
                    & OceanBase (MySQL)   & /             & allow             & allow             & /                 \\\cline{2-6}
                    & Greenplum   & allow             &    /          & allow             & allow                 \\\cline{2-6}
                    & PostgreSQL          & SSI(allow)             & allow             & allow             & allow             \\\cline{2-6}
                    & CockroachDB         & abort             & /                 & /                 & /                 \\\cline{2-6}
                    & MongoDB     & /                 & allow             & /                 & /     \\\hline
\bottomrule   
\end{tabular}

\end{table}

Unlike previous analyses that discussed only primitive conflicts, we, in this paper, explain more POPs. We exclude the discussion of RA, WA, and WC in most cases, as they (i) exist only in a 2-transaction single-object cycle and (ii) perform similar to RW, WW, and WW, respectively. With the Wait strategy, the second operation of primitive POPs waits for the first one to be committed, meaning WR, WW, and RW will turn into WCR, WCW, and RCW, respectively. We then discuss more detailed behaviors of WCR, WCW, and RCW. 

The \textbf{WCR} occurred when the committed write is read. 
There are two cases: After a transaction with the write operation is committed, other transactions can read the data. For example, the Read Skew Committed ($R_1[x_0]W_2[x_1]W_2[y_1]C_1R_1[y_1]$) was executed as expected by most databases (e.g., PostgreSQL, MySQL) at RC and RU levels. 
The WCR is formed by $W_2[y_1]C_1R_1[y_1]$ (compared to the Read Skew ($R_1[x_0]W_2[x_1]W_2[y_1]R_1[y_1]$), where WCR does not exist).
However, at the SER or RR level, by snapshot enabling (e.g., PostgreSQL), requiring to read a snapshot version $y_0$, Read Skew Committed was executed into a non-anomaly ($R_1[x_0]W_2[x_1]W_2[y_1]C_1$ $R_1[y_0]$), transforming WCR into RW. 

The \textbf{WCW} occurred when the write is allowed after the concurrent write is committed. The WCW is allowed in most cases but not allowed in Databases with only write locks (e.g., PostgreSQL and Oracle) at SER, SI, and RR levels. 
For example, the Dirty Write ($W_1[x_1]W_2[x_2]C_1$) was passed in MySQL, as it was executed into a non-anomaly ($W_1[x_1]C_1W_2[x_2]$), where only one WCW exists.  However, PostgreSQL aborted at SER and RR levels due to WCW POP. Similar cases are full Write and Full Write Committed.

The \textbf{RCW}  is very much the same behavior as RW and is allowed in all databases. 
For example, the Intermediate Read and Intermediate Read Committed having RW and RCW, respectively, performed quite the same at different levels. At SER, 2PL databases actually executed the Intermediate Read into the Intermediate Read Committed.
Similar cases happened between Read Skew 2 and Read Skew 2 Committed, between Read-write Skew 2 and Read-write Skew 2 Committed. One exception is that Oracle handled RW and RCW differently, as it passed Write Skew (with RW) but abort Write Skew Committed (with RCW).

In summary, at the SER level, 2PL databases (e.g., MySQL, SQL Server) do not allow WW, WR, and RW by 2PL Wait. However, they allow WCW, WCR, and RCW (as shown from SDA cases in Table \ref{table:evaluation_dbtest_result}).  
Other databases (e.g., PostgreSQL and Oracle) do not allow WW at all levels and do not allow WCW at SER and RR levels, yet allow all other POPs, while PostgreSQL (using SSI) did not allow two consecutive RWs. 
At weaker isolation levels, all databases still forbid WW but gradually allow more POPs like RW and WCR.

\begin{table}[t!]
	\caption{Anomalies at different isolation levels.}
	\label{table:evaluation_appeared_anomalies}
	\scriptsize
	\centering
	\begin{tabular}{c|c|c|c}
\toprule
\hline
No &  POP combinations & Example anomalies & Anomaly types\\\hline
1 &  RW & Write Skew, Step IAT & IAT \\\hline
2 & RW, RCW & Write Skew Committed & IAT\\\hline
3 & RW, WCW & Lost Update Committed & IAT \\\hline
4 & RW, WCR & Read Skew Committed & IAT \\\hline
5 & all but no WW & Read Skew, Write-read Skew & RAT, IAT\\\hline
\bottomrule  
\end{tabular}
\begin{tabular}{c|p{0.30in}|p{0.40in}|p{0.55in}|p{0.30in}}
\toprule
\hline

Databases & SER & RR      &   RC & RU      \\ \hline
MySQL/TDSQL	& None 	& 1. 2. 3. & 1. 2. 3.  4.   & 5.   \\\hline

SQL server & None 	& None	&   1. 2. 3.  4.    & 5. \\\hline
SQL server (SI) & /	& 1. 2.	&   1. 2. 3.  4.    & / \\\hline

TiDB	& /	&   1. 2. 3.	&  1. 2. 3.  4.    & /  \\\hline
TiDB (OPT)	& /	&   / 	&  1. 2. 3.   & /  \\\hline
Oracle	& 1.  2. 	&   /	& 1. 2. 3.  4.  & /  \\\hline
OceanBase (Oracle)	& 1. 2.	&  1. 2.	&  1. 2. 3.  4.  & /  \\\hline
OceanBase (MySQL)	& /	&  /	&  1. 2. 3.  4.  & /  \\\hline
Greenplum	& 1. 2.	&  /	&  1. 2. 3.  4.  & /  \\\hline
PostgreSQL	& None 	&   1. 2.	& 1. 2. 3. 4.  & 1. 2. 3.  4.    \\\hline
CockroachDB	& None 	&   / & / & / \\\hline
MongoDB (SS)	& /	&   1. 2.  & / & / \\\hline
\bottomrule                      
\end{tabular}

\end{table}

\subsubsection{\textbf{Data Anomalies Occurrence}}\label{sec:DataAnomalyOccurrence}
\
\newline
This part discusses occurrences of anomalies at different isolation levels. 
Table \ref{table:evaluation_appeared_anomalies} shows a summary of expected anomaly groups at different levels. We show 5 groups of different types of anomalies by different POP combinations. For example, Group 1 is the anomaly of any number of RW combinations. The typical anomalies are Write Skew and Step IAT in IAT. We found most databases allowed Group (1,2) or Group (1,2,3) at the RR level and allowed Group 4 furthermore at the RC level. While at the RU level, it allows anomalies formed by all POPs except WW.
We show a more detailed evaluation of anomaly occurrences from two perspectives, i.e., (i) expected performance that anomalies should appear and (ii) unexpected performance that anomalies should have been forbidden, in the following.

The \textbf{RAT} exists at least one WR POP.
(i) Based on our previous analysis, WR is indeed allowed only at the RU level by 2PL databases (e.g., MySQL and SQL Server). At the RU level, most schedules are executed as expected and anomalies are not detected. In contrast, at non-RU levels, RATs are mostly passed, as most of them are turned WR into RW by MVCC or WCR by 2PL Wait. 
For example, MySQL executed Intermediate Read ($W_1[x_1]R_2[x_1]W_1[x_2]$) as expected at the RU level but executed it into non-anomalies ($W_1[x_1]W_1[x_2]C_1R_2[x_2]C_2$) (WR to WCR by 2PL Wait) and ($W_1[x_1]$ $R_2[x_0]W_1[x_2]C_1C_2$) (WR to RW by MVCC) at SER and RR/RC levels, respectively. Interestingly, SQL Server executed Intermediate Read into one non-anomaly ($W_1[x_1]W_1[x_2]C_1R_2[x_2]C_2$) (WR to WCR by 2PL Wait) at all non-RU levels. 
(ii) RATs are not expected in non-RU levels, but some anomalies are reported, as they are executed into IATs.
For example, at the RR level, Most DB executed both Write-read Skew ($W_1[x_1]R_2[x_1]W_2[y_1]R_1[y_1]$) and Write-read Skew Committed ($W_1[x_1]R_2[x_1]W_2[y_1]C_2R_1[y_1]$) are often executed into Write Skew ($W_1[x_1]$ $R_2[x_0]W_2[y_1]$ $R_1[y_0]$), except SQL Server did not allow RW , ending up as a deadlock. 
However, MySQL executed Write-read Skew into a non-anomaly ($W_1[x_1]W_2[y_1]$ $R_2[x_0]C_2R_1[y_1]$), (due to the timing of taking snapshot, more details in \ref{sec:evaluation_mvcc_snapshot}).



The \textbf{WAT} exists at least one WW POP and without WR.
(i) WW is not allowed in all databases and at any level. For example, anomalies with all WWs like Full-write Skew, Full-write Skew Committed, and Step WAT are aborted in most databases at all levels. These anomalies are often detected as deadlocks (more detailed analysis in Section \ref{sec:evaluation_deadlocks}). 
(ii) However, we see some cases are passed. For example, Dirty Write ($W_1[x_1]W_2[x_2]C_1$) and Full-write anomalies were executed into non-anomalies (e.g., $W_1[x_1]C_1W_2[x_2]$ by Dirty Write) in most cases, transforming WW into WCW. Similar cases are Lost Self Update Committed, Double-write Skew 2 Committed, Read-write Skew 1/2, and Read-write Skew 2 Committed. However, some databases (e.g., PostgreSQL and Oracle), which disallowed WCW, aborted these two cases at SER, SI, OPT, and RR levels, but can execute the Dirty Write (abort version) ($W_1[x_1]W_2[x_2]A_1$) into a non-anomaly ($W_1[x_1]A_1W_2[x_2]C_2$). 

The \textbf{IAT} does not exist WR and WW. 
(i) Most databases tolerate IATs at the non-SER level. At the RR level, most databases occurred anomalies with RW or RCW combinations. The typical anomalies are Write Skew, Step IAT, etc. At RC and RU levels, they further occurred anomalies with WCW or WCR POPs. The typical anomalies are Lost Update Committed, Read Skew Committed, etc. 
(ii) Oracle, OceanBase, and Greenplum claimed to support SER level yet it behaves similar to RR or SI equivalent level. They eliminated four standard anomalies but ignored some anomalies in IAT. 
We further discuss the behaviors of OceanBase by Read Skew Committed, Read-write Skew 1 Committed, and Write Skew Committed anomalies, which are with RW-WCR, RW-WCW, and RW-RCW POP combinations, respectively. At the RC level, OceanBase executed these three anomaly schedules as expected, reporting anomalies. While at the SER/RR level, OceanBase behaved quite differently. OceanBase (1) passed Read Skew Committed due to snapshot reading, transforming WCR into RW, (2) aborted Read-write Skew 1 Committed due to WCW abort rule, and (3) reported an anomaly by Write Skew Committed as it executed as expected.


In summary, at the SER level, no anomalies occurred except for Oracle, OceanBase, and Greenplum. As most knowledge, researchers discover Oracle not to be consistent at its SER level by the Write Skew anomaly. We found that anomalies also happened by Write-read Skew (both committed version and non-committed version), Step RAT, and Step IAT, although they executed into Write Skew eventually.
At the RR level, most databases occurred anomalies with RW and RCW combinations (e.g., Read Skew, Write Skew, and Step RAT), except for SQL Server with a similar strong policy as SER level. Surprisingly, SQL Server has the same behaviors between SER and RR levels by our tests.  At the RC level, most databases occurred all anomalies happened at the RR level, and anomalies with RW and WCW/WCR combinations (e.g., Lost Update Committed and Read-write Skew 1 Committed). 
At the weakest RU level, all databases only avoided WW, resulting in all kinds of anomalies without WW (e.g., Read Skew and Read Skew2). Thus, most anomalies occur at the RU level. And among all types of anomalies, IAT are the trickiest one with RW and other POPs, having the most anomaly cases.


\vspace{0.1cm}
\noindent \textbf{Lesson learned:} \indent 
\noindent (i) Databases aim at consistency by avoiding all or partial POP cycles, and have different behaviors on different POPs. 
\noindent (ii)  
Different CC protocols are differently implemented between databases and between isolation levels. 
\noindent (iii) 
Developers still lack complete understanding between SER level and eliminating four standard anomalies, and between coarse isolation levels.
Our evaluation can capture more insights and subtle behaviors between POPs, CC, and coarse isolation levels.



\subsection{MVCC and Consistency}\label{sec:evaluation_mvcc_snapshot}
MVCC technology has three elements: multi versions, snapshot and data visibility algorithm.
Multi versions with Read committed rule allow the newest committed objects to be read at RC levels. It helps to transform WR into RW.
Snapshot, however, makes every read of the transaction consistent with exactly one committed version at SER and RR levels.
It transforms WR and WCR into RW. 
For example, PostgreSQL and MySQL passed Non-repeatable Read Committed ($R_1[x_0]W_2[X_1]C_2R_1[x_1]$) at RR levels as it executed into a non-anomaly ($R_1[x_0]W_2[X_1]C_2$ $R_1[x_0]$) but reported an anomaly at the RC level as expected. A similar case is Read Skew Committed.

MVCC sometimes are differently implemented. CockroachDB also consider Timestamp Ordering (TO) \cite{DBLP:journals/tods/BernsteinG83} in its CC protocols. Unlike traditional MVCC databases, reads are not waited/blocked, some read is waited in CockroachDB if an early uncommitted write found. For example, Write-read Skew Committed ($W_1[x_1]W_2[y_1]$ $R_2[x_1]C_2R_1[y_1]$) was executed into Write Skew by traditional MVCC databases like PostgreSQL, but into a non-anomaly ($W_1[x_1]W_2[y_1]$ $R_1[y_0]C_1R_2[x_1]C_2$) by the CockroachDB. Note that $T_1$ started earlier can read $y_0$, but $T_2$ started latter can not read $x_0$ but can read $x_1$ once $T_1$ committed.

Snapshot is the MVCC restricted to reading only one consistent version. 
Most databases (e.g., PostgreSQL, Oracle, and OceanBase 2.2.50) take the snapshot at the timestamp of first operation while some (i.e., MySQL and OceanBase 2.2.77) take the snapshot at the the first read. For example at the RR level, PostgreSQL executed the Write-read Skew Committed ($W_1[x_1]W_2[y_1]R_2[x_1]C_2R_1[y_1]$) into the Write Skew ($W_1[x_1]W_2[y_1]R_2[x_0]C_2R_1[y_0]$), printing anomaly found, as it takes the snapshot of $x_0$ and $y_0$. However, MySQL executed Write-read Skew Committed into a non-anomaly ($W_1[x_1]$ $W_2[y_1]R_2[x_0]C_2R_1[y_1]$), as it takes the snapshot of $x_0$ and $y_1$ at the first read. 

\vspace{0.1cm}
\noindent \textbf{Lesson learned:} \indent
MVCC helps to transform WR into RW, and snapshot transforms WR and WCR into RW. Most databases (e.g., PostgreSQL) take snapshots at the beginning of the transaction while some (e.g., MySQL) at the first read.
\subsection{Distributed Consistency}\label{sec:evaluation_distributed_consistency}
The above analyses are based on the centralized evaluation. This part discusses the evaluation of distributed databases.
We deployed the data to be stored in different partitions/nodes. In Greenplum, as a write by default has a lock on one table/segments, we can distribute each row/keys to be in different tables. Note that SDAs (e.g., four standard anomalies) with one object are not suitable for the distributed consistency check. We want to observe the difference from the global CC protocols and deadlock detection. 
We evaluated 5 databases (i.e., MongoDB, CockroachDB, Greenplum, OceanBase, and TiDB) by DDAs and MDAs.
We obtained the same results as Table \ref{table:evaluation_dbtest_result}, meaning these databases did a great implementation to maintain the consistent performance between centralized and distributed deployments. 

We showcase the Write Skew ($R_1[x_0]R_2[y_0]W_2[x_1]W_1[y_1]$) anomaly occurred in distributed scenario by Greenplum. We let objects $x$ and $y$ be stored in two tables on two partitions. And then the Write Skew was executed as scheduled at the SER level, meaning an anomaly is found. Similar cases are Write Skew committed, Write-read Skew, etc. OceanBase at the SER level, MongoDB at the SI level, and TiDB at the RR level, existed similar anomalies in distributed scenarios.

\noindent \textbf{Lesson learned:} 
DDA and WDA type anomalies are suitable for distributed environments. CockroachDB has excellent consistent behaviors at the SER level as the centralized scenarios. But OceanBase and GreenPlum are not.





\subsection{Deadlocks}\label{sec:evaluation_deadlocks}





Deadlocks occurred when multiple transactions wait for each other for resources.  Deadlocks are usually found by periodically checking wait-for graphs \cite{DBLP:conf/sigmod/LyuZXGWCPYGWLAY21}. 
Most databases (e.g., PostgreSQL, CockroachDB, and Oracle) use deadlock detection only for a small portion of data anomalies, and they detect deadlocks from anomalies Full-write Skew, Full-write Skew Committed, and Step WAT by two or three WW POPs waiting, as they have only locks on writes. In contrast, 2PL databases (e.g., SQL Server and MySQL) heavily detect deadlocks from all rollbacked cases, as they may have locks on both reads and writes, making WR, WW, and RW POPs wait for each other. Table \ref{table:evaluation_test_case_read_skew}(D) depicts an example of Step WAT rollbacked by PostgreSQL deadlock detection. 
The transaction which found deadlock is often aborted and the rest may continue to proceed. However, in PostgreSQL and Oracle, the transaction which found deadlock aborted while the rest are still waiting. They by default can not proceed and depend on \textit{lock\_timeout} to terminate. And OceanBase did not use any deadlock techniques at all, instead, it used timeout (e.g., 2PL Wait\_die) to avoid deadlocks.

\vspace{0.1cm}
\noindent \textbf{Lesson learned:} \indent 
(i) Deadlocks are caused by resource wait-for dependency by the 2PL Wait. 
(ii) Deadlocks are essentially special instances of data anomalies. 



\section{Related Work}\label{sec:related_work}
In this part, we surveyed the related work in more detail.



\vspace{0.1cm}
\noindent \textbf{Consistency} \indent
For database transactions, there are two classical definitions of C in ACID. First, the ANSI SQL \cite{ansisql} holds that consistency is met without violating integrity constraints; Second, Jim \cite{DBLP:conf/ds/GrayLPT76} believes that consistency can be divided into four levels, and each level excludes some data anomalies. Both of them are casual definitions and cannot directly and specifically guide the consistency verification of the database. 
Some \cite{Jepsen-Hermitage,Jepsen-reports} reported that many databases do not provide the consistency and isolation guarantees they claimed. 
In fact, within the scope of the database, there is little research on the definition of consistency, not to mention the research on the relationship between consistency and data anomalies. 
Adya et al. \cite{adya1999weak} defines the relationship between conflict graphs and data anomalies. However, they can not correspond to some kinds of data anomalies (e.g., Dirty Read, Dirty Write and Intermediate Read \cite{ansisql, DBLP:conf/sigmod/BerensonBGMOO95}). The reason behind this is that the stateful information like commit and abort cannot be modeled in the conflict graph.
In contrast, this paper proposed a POP graph that can fully express the schedule with this stateful information. By POP graph, we are able to define all data anomalies and corresponding consistency to no anomalies.

\vspace{0.1cm}
\noindent \textbf{Consistency check} \indent 
There exist two typical methods for checking databases consistency. One is by the white box method \cite{DBLP:conf/pldi/XuBH05,DBLP:conf/icse/HammerDVT08,DBLP:conf/ipps/SinhaM10,DBLP:conf/rv/SumnerHD11,DBLP:journals/vldb/ZellagK14,DBLP:conf/popl/BrutschyD0V17,DBLP:conf/concur/NagarJ18}, where users often profile active transactions and conflicts to check non-serializable schedule. The white box method has a high knowledge bar and user-side burden to modify system code. As active transactions increases, the checking cost may exponentially increase, possibly affecting the performance of original transaction processing. The other is by the black box method \cite{DBLP:journals/pacmpl/BiswasE19,DBLP:conf/osdi/TanZMW20}, where users do not make any modification for the system and check the result by some given workloads. Jepsen (including Elle \cite{DBLP:journals/pvldb/AlvaroK20}, which is part of the Jepsen project) consistency check \cite{Jepsen-Hermitage} is one of the popular tools in the industry. However, these methods usually issue random workloads to discover inconsistent behaviors. Such methods are not accurate, spending tons of computing resources. In contrast, Coo judiciously designed finite anomaly schedules, evaluating the consistency of databases once and for all. The evaluation is accurate (all types of anomalies), user-friendly (SQL-based test), and cost-effective (few minutes). The test is also possible for distributed databases. Test cases (i.e., DDA and MDA, which have more than one object) can be designed to force data to spread in different partitions.
\vspace{0.1cm}
\noindent \textbf{Data anomalies, serializability, and consistency} \indent
In recent years, there still exist extensive research works that focus on reporting new data anomalies, we make a thorough survey on data anomalies and show them in Table \ref{tab:intro_data_anomaly}. These new data anomalies are constantly reported in different scenarios, indicating that data consistency in various scenarios is still full of challenges. 
The traditional knowledge has a shallow and inaccurate understanding between data anomalies and consistency. Previous work related conflict acyclic graph to consistency. They guarantee the serializable schedule to guarantee the consistency \cite{DBLP:books/daglib/0006733, DBLP:books/aw/BernsteinHG87, 1702620, DBLP:conf/eurosys/YabandehF12}. The serialization is usually achieved by strong rules via eliminating three kinds of conflict relations (i.e., WW, WR, and RW) \cite{DBLP:books/daglib/0006733}. However, they can not quantify all data anomalies such as Dirty Read and Dirty Write. In this paper, Coo,  by using the POP graph, can define all anomalies and correlate data anomalies to inconsistency.

\section{Conclusion and future work}\label{sec:conclusion}


This paper proposed Coo, which contributed to pre-check the consistency of databases, filling the gap in contrast to real-time or post-verify solutions. We systematically defined all data anomalies and correlated data anomalies to inconsistency. 
Specifically, we introduced an extended conflict graph model called Partial Order Pair (POP) Graph, which also considers state-expressed operations. By POP cycles, we can produce infinite distinct data anomalies. We classify data anomalies and report 20+ new types of them. 
We evaluated the new consistency model by ten real databases. 
The consistency check by predefined representative anomaly cases is accurate (all types of anomalies), user-friendly (SQL-based test), and cost-effective (one-time checking in a few minutes).

The research of predicate cases have not been discussed in this paper due to the limited space and is still on going. We think the model of this paper is compatible to extend to predicate cases 
(e.g., Phantom can be constructed by Non-repeatable Read, with predicate Select and replacing Update by Insert \cite{coo_consistency_check}).

\bibliographystyle{ACM-Reference-Format}
\bibliography{reference}



\end{document}